%% file: TR0155.tex
\documentclass[a4paper,10pt]{llncs}

\pdfoutput=1

\usepackage{graphicx}
\usepackage{algo}
\usepackage{amsmath} %is automatically loaded by MnSymbol, 2011-01-19 rmi
\usepackage{amssymb} %incompatible with MnSymbol
\usepackage[center]{subfigure}
\usepackage{url}
\usepackage{xcolor}
\usepackage{verbatim}
\usepackage{listings}
\usepackage{multirow}
\usepackage{tikz}

\usepackage[ansinew]{inputenc}   % for Windows
\newcommand{\cV}{{\cal V}} %set of shared variables
\newcommand{\cL}{{\cal L}} %set of labels
\newcommand{\cM}{{\cal M}} %set of matrices
\newcommand{\cT}{{\cal T}} %set of matrices representing threads
\newcommand{\cS}{{\cal S}} %set of matrices representing semaphores
\newcommand{\cP}{{\cal P}} %program P

\newcommand\NUMSV[1]{\mbox{NSV}(#1)}

\newcommand\SV{\mbox{\scriptsize SV}} %index for set of lables (shared variables only)
\newcommand\NSV{\mbox{\scriptsize NSV}} %number shared variables on a edge or in a basic block
\newcommand\SP{\mbox{\scriptsize S}} %index for set of lables (sync points)
\newcommand\V{\mbox{\scriptsize V}}  %index for set of lables (all shared variable and non-shared variable labels, without sync points)

\newcommand\merry{\circ}

\newcommand{\bigo}[1]{ O\bigr( #1 \bigl) }

\newcommand\SumFromTo[3]{\overset{#2}{\underset{#1}{\sum}} #3}
\newcommand\SumFromToText[3]{{\sum}^{#2}_{#1} #3}

\newcommand\ProdFromToText[3]{\prod_{#1}^{#2}\, #3}

\newcounter{tbdcounter}
\begin{document}

    \pagestyle{plain} %'plain' displays page number in the footer
    \pagenumbering{arabic}

    \title{Shared Memory Concurrent System Verification using Kronecker Algebra}

    \subtitle{\small Technical Report 183/1-155}

    \author{Robert Mittermayr \and Johann Blieberger}

    \institute{Institute of Computer-Aided Automation, TU Vienna, Austria}

    \maketitle

    \begin{abstract}
        The verification of multithreaded software is still a challenge. This comes mainly from the fact that the number of thread interleavings grows exponentially in the number of threads. The idea that thread interleavings can be studied with a matrix calculus is a novel approach in this research area. Our sparse matrix representations of the program are manipulated using a lazy implementation of Kronecker algebra. One goal is the generation of a data structure called \emph{Concurrent Program Graph} (CPG) which describes all possible interleavings and incorporates synchronization while preserving completeness. We prove that CPGs in general can be represented by sparse adjacency matrices. Thus the number of entries in the matrices is linear in their number of lines. Hence efficient algorithms can be applied to CPGs. In addition, due to synchronization only very small parts of the resulting matrix are actually needed, whereas the rest is unreachable in terms of automata. Thanks to the lazy implementation of the matrix operations the unreachable parts are never calculated. This speeds up processing significantly and shows that this approach is very promising.

        Various applications including data flow analysis can be performed on CPGs. Furthermore, the structure of the matrices can be used to prove properties of the underlying program for an arbitrary number of threads. For example, deadlock freedom is proved for a large class of programs.
    \end{abstract}
%
    %\footnotetext{\today}
    %
    %for sigplan doc class
    %\category{D.2.4}{Software Engineering}{Software/Program Verification}%[Correctness proofs]
    %\category{F.3.1}{Logics and Meanings of Programs}{Specifying and Verifying and Reasoning about Programs}
    %\category{F.3.2}{Logics and Meanings of Programs}{Semantics of Programming Languages}[Program analysis]
%
    %\terms
    %Verification, Concurrency
%
    %\keywords
    %Shared Memory Concurrent Systems, Interleavings Semantics, Deadlocks, Synchronization, Sequentialization, Kronecker Algebra
%
    \input{introduction.tex}
    \input{preliminaries.tex}
    \input{model.tex}
    \input{exampleCS.tex}
    \input{deadlock.tex}
    \input{example.tex}
    \input{empirical.tex}
    \input{related.tex}
    \input{conclusion.tex}
\input{references.tex}
\end{document}

%% file: introduction.tex
\begin{section}{Introduction}\label{section_introduction}
    %For safety-critical systems, dependable systems or robust embedded systems, software has to be provably correct. In particular, this is a very important issue in the fields of medical systems, aviation, rail, and automotive industries.
%
    %\TBD{should we move the client-server example to Subsect.~\ref{subsection_lazyImpl}?}\\
    %\TBD{should we cite Lipton?, No}\\
    %\TBD{should we cite TACAS'08 paper?, No}\\
    %\TBD{ignore this comment: ``Sync. over shared variables''?}\\
    %\TBD{TACAS: remove forward references}\\
    %\TBD{TACAS: explain false positives for deadlocks in a sentence, or should we write an example?}\\
    %\TBD{TACAS: concerning the different views, matrix, automata \dots, should we clarify somewhere?, where should we state that we only work on matrices}\\
    %\TBD{TACAS: explain lazy approach in more detail}\\
    With the advent of multi-core processors scientific and industrial interest focuses on the verification of multithreaded applications. The scientific challenge comes from the fact that the number of thread interleavings grows exponentially in a program's number of threads. All state-of-the-art methods, such as model checking, suffer from this so-called {\em state explosion problem}.
%
    %for reviewing self references in 3rd person!
    %In~\cite{MB:08} we started research in order to find where the explosion has its origin. Although conditionals and loops cannot be analyzed with this early approach it anyway indicated the high prospects in terms of interleaving reduction compared to the theoretic amount of interleavings.
%
    %With the matrix calculus based approach presented in this paper we are able to support conditionals, loops, and synchronization-sensitiveness.
    The idea that thread interleavings can be studied with a matrix calculus is new in this research area. We are immediately able to support conditionals, loops, and synchronization. Our sparse matrix representations of the program are manipulated using a lazy implementation of Kronecker algebra. Similar to~\cite{BK:02} we describe synchronization by Kronecker products and thread interleavings by Kronecker sums.
    One goal is the generation of a data structure called \emph{Concurrent Program Graph} (CPG) which describes all possible interleavings and incorporates synchronization while preserving completeness. Similar to CFGs for sequential programs, CPGs may serve as an analogous graph for concurrent systems. We prove that CPGs in general can be represented by sparse adjacency matrices. Thus the number of entries in the matrices is linear in their number of lines.

    In the worst-case the number of lines increases exponentially in the number of threads. Especially for concurrent programs containing synchronization this is very pessimistic. For this case we show that the matrix contains nodes and edges unreachable from the entry node.

    We propose two major optimizations. First, if the program contains a lot of synchronization, only a very small part of the CPG is reachable. Our lazy implementation of the matrix operations computes only this part (cf. Subsect.~\ref{subsection_lazyImpl}). Second, if the program has only little synchronization, many edges not accessing shared variables will be present, which are reduced during the output process of the CPG (cf. Subsect.~\ref{subsection_optimizationForNSV}). Both optimizations speed up processing significantly and show that this approach is very promising.

    We establish a framework for analyses of multithreaded shared memory concurrent systems which forms a basis for analyses of various properties. Different techniques including dataflow analysis (e.g.~\cite{RP:86,RP:88,SGL:98,KU:76})
    %symbolic analysis (e.g.~\cite{BSB:08}),
    and model checking (e.g.~\cite{CGP:99,GG:08} to name only a few) can be applied to the generated \emph{Concurrent Program Graphs} (CPGs) defined in Section~\ref{section_CPG}. Furthermore, the structure of the matrices can be used to prove properties of the underlying program for an arbitrary number of threads. For example in this paper, deadlock freedom is proved for p-v-symmetric programs.

    Theoretical results such as~\cite{Ram:00} state that synchronization-sensitive and con\-text-sensitive analysis is impossible even for the simplest analysis problems. Our system model differs in that it supports subprograms only via inlining and recursions are impossible.

    The outline of our paper is as follows. In Section~\ref{section_prelim} control flow graphs, edge splitting, and Kronecker algebra are introduced. Our model of concurrency, its properties, and important optimizations like our lazy approach are presented in Section~\ref{section_CPG}. In Section~\ref{section_exampleCS} we give a client-server example with 32 clients showing the efficiency of our approach. For a matrix with a potential order of $10^{15}$ our lazy approach delivers the result in 0.43s. Section~\ref{section_deadlock} demonstrates how deadlock freedom is proved for p-v-symmetric programs with an arbitrary number of threads. An example for detecting a data race is given in Section~\ref{section_example}. Section~\ref{section_empiricalData} is devoted to an empirical analysis. In Section~\ref{section_related} we survey related work. Finally, we draw our conclusion in Section~\ref{section_conclusion}.
\end{section} 

%% file: preliminaries.tex
\begin{section}{Preliminaries}\label{section_prelim}
    %we do this later \TBD{semantics for shared-memory concurrent systems (SMCS), need proof that what we generate is semantically equivalent to SMCS.}
%
    \begin{subsection}{Overview}\label{subsection_overview}
        %{\em static}, {\em sound}, {\em fully automated}, {\em scalable},
        We model shared memory concurrent systems by threads which use semaphores for synchronization.
        Threads and semaphores are represented by control flow graphs (CFGs). Edge Splitting has to be applied to the edges of thread CFGs that access more than one shared variable. Edge splitting is straight forward and is described in Subsect.~\ref{subsection_edgeSplitting}. The resulting Refined CFGs (RCFGs) are represented by adjacency matrices. These matrices are then manipulated by Kronecker algebra.
        We assume that the edges of CFGs are labeled by elements of a semiring. Details follow in this subsection. Similar definitions and further properties can be found in~\cite{KS:86}.
%
        %Let $\Sigma$ be a finite alphabet which consists of so-called \emph{letters} or \emph{symbols}. The set of all words constructed out of symbols is denoted by $\Sigma^*$. A subset of $\Sigma^*$ is being referred to as a formal language over $\Sigma$.

        Semiring $\langle \cL,+,\cdot,0,1 \rangle$ consists of a set of labels $\cL$, two binary operations $+$ and $\cdot$, and two constants $0$ and $1$ such that
        \begin{enumerate}
            \item $\langle \cL,+,0 \rangle$ is a commutative monoid,
            \item $\langle \cL,\cdot,1 \rangle$ is a monoid,
            \item $\forall l_1, l_2, l_3 \in \cL: l_1 \cdot (l_2 + l_3) = l_1 \cdot l_2 + l_1 \cdot l_3$ and $(l_1 + l_2) \cdot l_3 = l_1 \cdot l_3 + l_2 \cdot l_3$ hold and
            \item $\forall l \in \cL: 0 \cdot l = l \cdot 0 = 0$.
        \end{enumerate}

        Intuitively, our semiring is a unital ring without subtraction. %A semiring is called commutative if $\forall a,b \in S: a \cdot b=b \cdot a$. If $\forall a \in S: a+a=a$ then it is called idempotent.
        For each $l\in \cL$ the usual rules are valid, e.g., $l + 0 = 0 + l = l$ and $1 \cdot l=l \cdot 1 = l$.
        In addition we equip our semiring with the unary operation $*$. For each $l \in \cL$, $l^*$ is defined by $l^*=\SumFromToText{j \geq 0}{}{l^j}$, where $l^0=1$ and $l^{j+1}=l^j \cdot l = l \cdot l^j$ for $j \geq 0$.
        Our set of labels $\cL$ is defined by $\cL=\cL_{\V} \cup \cL_{\SP}$, where $\cL_{\V}$ is the set of non-synchronization labels and $\cL_{\SP}$ is the set of labels representing semaphore calls. The sets $\cL_{\V}$ and $\cL_{\SP}$ are disjoint.
        The set $\cL_{\SP}$ itself consists of two disjoint sets $\cL_{\SP_p}$ and $\cL_{\SP_v}$. The first denotes the set of labels referring to P-calls, whereas the latter refers to V-calls of semaphores.

        Examples for semirings include regular expressions (cf.~\cite{Tar:81}) which can be used for performing dataflow analysis.
%
        %The \emph{sum star identity} is valid for $a$ and $b$ if $(a+b)^*=(a^* \cdot b)^* \cdot a^*$. The \emph{product star identity} is valid for $a$ and $b$ if $(a \cdot b)^*=1+a \cdot (b \cdot a)^* \cdot b$. A star semiring in which sum star identity and product star identity are valid is called \emph{Conway semiring}. Further the \emph{star fixed point identity} $a^*=1+a \cdot a^*$ and the \emph{simplified product star identity} $a \cdot (b \cdot a)^*=(a \cdot b)^* \cdot a$ hold for all semiring elements $a$ and $b$.
    \end{subsection}

    \begin{subsection}{Control Flow Graphs}\label{subsection_controlFlowGraphs}
        A \emph{Control Flow Graph} (CFG) is a directed labeled graph defined by $G=\langle V,E,n_e \rangle$ with a set of nodes $V$, a set of directed edges $E~\subseteq~V~\times~V$, and a so-called \emph{entry} node $n_e \in V$. We require that each $n \in V$ is reachable through a sequence of edges from $n_e$. Nodes can have at most two outgoing edges. Thus the maximum number of edges in CFGs is $2 \, |V|$. We will use this property later.

        Usually CFG nodes represent basic blocks (cf.~\cite{ASU:86}). Because our matrix calculus manipulates the edges we need to have basic blocks on the edges.\footnote{We chose the incoming edges.}
        %Thus we have the \emph{basic blocks} situated on the incoming edges.
        Each edge $e \in E$ is assigned a basic block $b$. In this paper we refer to them as edge labels as defined in the previous subsection. To keep things simple we use edges, their labels and the corresponding entries of the adjacency matrices synonymously.

        In order to model synchronization we use semaphores. The corresponding edges typically have labels like $p_1$ and $v_1$, where $p_x$ and $v_x \in \cL_{\SP}$. Usually two or more distinct thread CFGs refer to the same semaphore to perform synchronization. The other labels are elements from $\cL_{\V}$. The operations on the basic blocks are $\cdot, +$, and $*$ from the semiring defined above (cf.~\cite{Tar:81}). Intuitively, $\cdot, +$, and $*$ model consecutive basic blocks, conditionals, and loops, respectively.

        \begin{figure}[t]
            \centering
            \subfigure[Binary Semaphore]{
                \hspace{13mm}\includegraphics[scale=0.4]{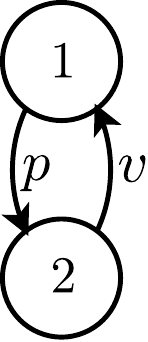}\hspace{13mm}
                \label{figure_semaphoreBinary}
            }
            \subfigure[Counting Semaphore]{
                \hspace{13mm}\includegraphics[scale=0.4]{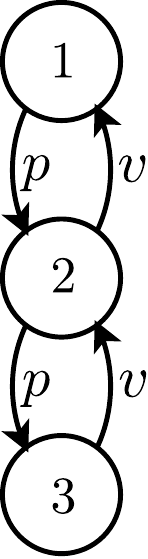}\hspace{13mm}
                \label{figure_semaphore2}
            }
            \caption{Semaphores}
            \label{figure_semaphores}
        \end{figure}

        In Fig.~\ref{figure_semaphoreBinary} and ~\ref{figure_semaphore2} a binary and a counting semaphore are depicted. The latter allows two threads to enter at the same time. In a similar way it is possible to construct semaphores allowing $n$ non-blocking P-calls.
    \end{subsection}

    \begin{subsection}{Edge Splitting}\label{subsection_edgeSplitting}
        A basic block consists of multiple consecutive statements without jumps. For our purpose we need a finer granularity as we would have with basic blocks alone. To achieve the required granularity we need to split edges. Shared variable accesses and semaphore calls may occur in basic blocks. For both it is necessary to split edges. This ensures a representation of possible context switches in a manner exact enough for our purposes.
        We say ``exact enough'' because by using basic blocks together with the above refinement, we already have coarsened the analysis compared to the possibilities on statement-level. Furthermore we do not lose any information required for the completeness of our approach.
        Anyway, applying this procedure to a CFG, i.e. splitting edges in a CFG, results in a \emph{Refined Control Flow Graph} (RCFG).

        Let $\cV$ be the set of shared variables. In addition, let a shared variable $v \in \cV$ be a volatile variable located in the shared memory which is accessed by two or more threads.
        Splitting an edge depends on the number of shared variables accessed in the corresponding basic block. For edge $e$ this number is being referred to as $\NUMSV{e}$. In the same way we refer to $\NUMSV{b}$ as the number of shared variables accessed in basic block $b$. If $\NUMSV{e}>1$, edge splitting has to be applied to edge $e$; the edge is used unchanged otherwise.
        %Similarly we define \SVPRED{e} and \SVPRED{b} to refer to the accessed shared variable itself.

        If edge splitting has to be applied to edge $e$ which has basic block $b$ assigned and $\NUMSV{b}=k$ then the basic blocks $b_1, \dots, b_k$ represent the subsequent parts of $b$ in such a way that $\forall b_i: \NUMSV{b_i}=1$, where $1\leq i \leq k$. Edges $e_j$ get assigned basic block $b_j$, where $1 \leq j \leq k$. In Fig.~\ref{figure_edgeSplitting} the splitting of an edge with basic block $b$ and $\NUMSV{b}=k$ is depicted.

        \begin{figure}[ht]
            \centering
            \includegraphics[scale=0.4]{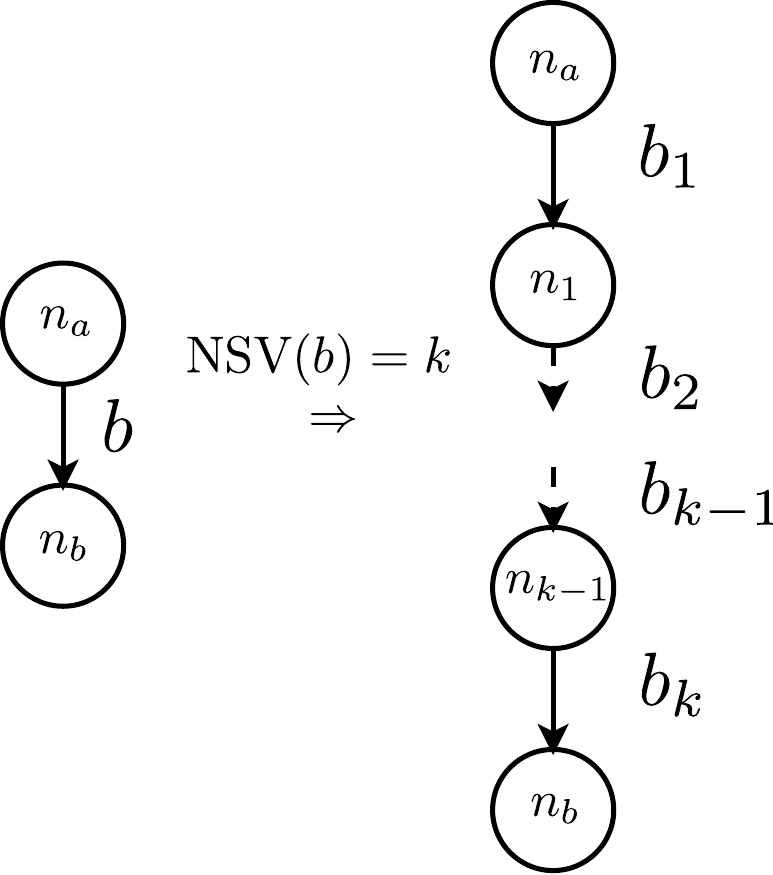}
            \caption{Edge Splitting for Shared Variable Accesses}
            \label{figure_edgeSplitting}
        \end{figure}

        For semaphore calls (e.g. $p_1$ and $v_1$) edge splitting is required in a similar fashion. In contrast to shared variable accesses we require that semaphore calls have to be the only statement on the corresponding edge. The remaining consecutive parts of the basic block are situated on the previous and succeeding edges, respectively.\footnote{Note that edges representing a call to a semaphore are not considered to access shared variables.}

        The effects of edge splitting for shared variables and semaphore calls can be seen in the data race example given in Section~\ref{section_example}. Each RCFG depicted in Fig.~\ref{figure_exampleRCFGs} is constructed out of one basic block (cf. Fig.~\ref{figure_exampleProgram}).

        Note that edge splitting ensures that we can model the minimal required context switches. The semantics of a concurrent programming language allows usually more. For example consider an edge in a RCFG containing two consecutive statements, where both do not access shared variables. A context switch may happen in between. However, this additional interleaving does not provide new information. Hence our approach provides the minimal number of context switches.
        %that it is semantically equivalent to more context switches.

        Without loss of generality we assume that the statements in each basic block are atomic. Thus, we assume while executing a statement, context switching is impossible. In RCFGs the finest possible granularity is at statement-level. If, according to the program's semantic,  atomic statements may access two or more shared variables, then we make an exception to the above rule and allow two or more shared variable accesses on a single edge. Such edges have at most one atomic statement in their basic block. The Kronecker sum (which is introduced in the next subsection) ensures that all interleavings are generated correctly.
    \end{subsection}

    \begin{subsection}{Synchronization and Generating Interleavings with Kronecker Algebra}\label{subsection_KroneckerAlgebraIntroduction}
        %Tensor product and tensor sum together form tensor algebra. We apply this algebra to RCFG adjacency matrices. Applied to matrices the tensor algebra is often called Kronecker algebra. The tensor product is then being referred to as \emph{Kronecker product} or \emph{Zehfuss product} and the tensor sum is called \emph{Kronecker sum}.

        Kronecker product and Kronecker sum form Kronecker algebra. In the following we define both operations, state properties, and give examples. In addition, for the Kronecker sum we prove a property which we call \emph{Mixed Sum Rule\/}.

        We define the set of matrices $\cM = \{ M = (m_{i,j}) \, | \, m_{i,j} \in \cL \}$. In the remaining parts of this paper only matrices $M \in \cM$ will be used, except where stated explicitly. Let $o(M)$ refer to the order\footnote{A k-by-k matrix is known as square matrix of order $k$.} of matrix $M \in \cM$. In addition we will use n-by-n zero matrices $Z_n=(z_{i,j})$, where $\forall i,j: z_{i,j}=0$.

        \begin{definition}[Kronecker product]
            Given a m-by-n matrix $A$ and a p-by-q matrix $B$, their \emph{Kronecker product} denoted by $A \otimes B$ is a mp-by-nq block matrix defined by
            \begin{equation*}
            A \otimes B =
                \begin{pmatrix}
                    a_{1,1}B & \cdots & a_{1,n}B \\
                     \vdots & \ddots & \vdots \\
                    a_{m,1}B & \cdots & a_{m,n}B
                \end{pmatrix}.
            \end{equation*}
        \end{definition}
        \begin{example}\label{example_kroneckerProduct}\mbox{}\\
            Let $A= \begin{pmatrix}
                        a_{1,1} & a_{1,2} \\
                        a_{2,1} & a_{2,2} \\
                    \end{pmatrix}$
            and
            $B= \begin{pmatrix}
                     b_{1,1} & b_{1,2} & b_{1,3} \\
                     b_{2,1} & b_{2,2} & b_{2,3} \\
                     b_{3,1} & b_{3,2} & b_{3,3} \\
                 \end{pmatrix}
            $.
            The Kronecker product $C = A \otimes B$ is given by\\

            $$
            \begin{pmatrix}
                \ a_{1,1}b_{1,1}\ & a_{1,1}b_{1,2}\ & a_{1,1}b_{1,3}\ & a_{1,2}b_{1,1}\ & a_{1,2}b_{1,2}\ & a_{1,2}b_{1,3}\ \\
                \ a_{1,1}b_{2,1}\ & a_{1,1}b_{2,2}\ & a_{1,1}b_{2,3}\ & a_{1,2}b_{2,1}\ & a_{1,2}b_{2,2}\ & a_{1,2}b_{2,3}\ \\
                \ a_{1,1}b_{3,1}\ & a_{1,1}b_{3,2}\ & a_{1,1}b_{3,3}\ & a_{1,2}b_{3,1}\ & a_{1,2}b_{3,2}\ & a_{1,2}b_{3,3}\ \\
                \ a_{2,1}b_{1,1}\ & a_{2,1}b_{1,2}\ & a_{2,1}b_{1,3}\ & a_{2,2}b_{1,1}\ & a_{2,2}b_{1,2}\ & a_{2,2}b_{1,3}\ \\
                \ a_{2,1}b_{2,1}\ & a_{2,1}b_{2,2}\ & a_{2,1}b_{2,3}\ & a_{2,2}b_{2,1}\ & a_{2,2}b_{2,2}\ & a_{2,2}b_{2,3}\ \\
                \ a_{2,1}b_{3,1}\ & a_{2,1}b_{3,2}\ & a_{2,1}b_{3,3}\ & a_{2,2}b_{3,1}\ & a_{2,2}b_{3,2}\ & a_{2,2}b_{3,3}\ \\
            \end{pmatrix}.
            $$

        \end{example}
        As stated in~\cite{EUWM:10} the Kronecker product is also being referred to as \emph{Zehfuss product\/} or \emph{direct product of matrices} or \emph{matrix direct product}.
        \footnote{
        Knuth notes in~\cite{Knu:11} that Kronecker never published anything about it. Zehfuss was actually the first publishing it in the $19$th century~\cite{Zeh:1858}. He proved that $\hbox{det}(A \otimes B)=\hbox{det}^n (A) \cdot \hbox{det}^m (B)$, if $A$ and $B$ are matrices of order $m$ and $n$ and entries from the domain of real numbers, respectively. % Although the main focus was on a determinant.
        }

        In the following we list some basic properties of the Kronecker product.
        Proofs and additional properties can be found in~\cite{Bel:97,Gra:81,Dav:81,Hur:1894}.
        %Let $k$ be a scalar and $A$, $B$, $C$ and $D$ matrices.
        Let $A$, $B$, $C$, and $D$ be matrices.
        The Kronecker product is noncommutative because in general $A \otimes B \neq B \otimes A$. It is permutation equivalent because there exist permutation matrices $P$ and $Q$ such that $A \otimes B = P (B \otimes A) Q$. If $A$ and $B$ are square matrices, then $A \otimes B$ and $B \otimes A$ are even permutation similar, i.e., $P = Q^T$. The product is associative as
        \begin{equation}\label{equation_kroneckerProduct_associativity}
            A \otimes (B \otimes C)=(A \otimes B) \otimes C.
        \end{equation}
        In addition, the Kronecker product distributes over $+$, i.e.,
        \begin{eqnarray}
             A \otimes (B+C) & = & A \otimes B + A \otimes C \label{equation_KroneckerProduct_bilinearity1},\\
             (A+B) \otimes C & = & A \otimes C + B \otimes C \label{equation_KroneckerProduct_bilinearity2}.%\\ % proof in Gra:81
        %     (kA) \otimes B  & = & A \otimes (kB) = k(A\otimes B). % proof in Gra:81
        \end{eqnarray}
        Hence for example $(A + B) \otimes (C + D) = A \otimes C + B\otimes C + A \otimes D + B\otimes D$.

        The Kronecker product allows to model synchronization (cf. Subsect.~\ref{subsection_proofSyncInOurModel}).

        %(info can be found in Hurwitz product /in die tiefe),
        %website: Earliest Known Uses of Some of the Words of Mathematics: (year of Hensel's Publication is wrong on the website!!!)
                %The KRONECKER, ZEHFUSS or DIRECT PRODUCT of matrices. In their article On the history of the kronecker product Linear and Multilinear Algebra, 14, (1983), 113  120, H. V. Jemderson, F. Pukelsheim & S. R. Searle find the origins of such products in a paper on determinants by Johann Georg Zehfuss Ueber eine gewisse Determinante, Zeitschrift für Mathematik und Physik, 3, (1858), 298-301. The link to Kronecker was made by K. Hensel, Ueber Gattungen, welche durch Composititon aus zwei anderen Gattungen entstehen Journal für die reine und angewandte Mathematik, 105, (1889), 329-344 who held that Kronecker discussed them in his lectures in the 1880s.
        %wapedia: "The Kronecker product of matrices corresponds to the abstract tensor product of linear maps."
        %         When V and W are Lie algebras, and S : V --> V and T : W --> W are Lie algebra homomorphisms,
        %         the Kronecker sum of A and B represents the induced Lie algebra homomorphisms V ? W --> V ? W.
        %         The Kronecker product of the adjacency matrices of two graphs is the adjacency matrix of the tensor product graph.
        \begin{definition}[Kronecker sum]\label{definition_kroneckerSum}
            Given a matrix $A$ of order $m$ and matrix $B$ of order $n$, their \emph{Kronecker sum} denoted by $A \oplus B$ is a matrix of order $mn$ defined by
            \begin{equation*}
                A \oplus B = A \otimes I_n + I_m \otimes B,
            \end{equation*}
            where $I_m$ and $I_n$ denote identity matrices\footnote{The identity matrix $I_n$ is a n-by-n matrix with ones on the main diagonal and zeros elsewhere.} of order $m$ and $n$, respectively.
        \end{definition}
        This operation must not be confused with the direct sum of matrices, group direct product or direct product of modules for which the symbol $\oplus$ is used too.
        By calculating the Kronecker sum of the adjacency matrices of two graphs the adjacency matrix of the Cartesian product graph~\cite{IKR:08} is computed (cf.\,~\cite{Knu:11}).
        \begin{example}\label{example_kroneckerSum}
            We use matrices $A$ and $B$ from Ex.~\ref{example_kroneckerProduct}. The Kronecker sum $A \oplus B$ is given by
            \begin{eqnarray*}
            %$$
                &&A \otimes I_3 + I_2 \otimes B =\\
                &&\begin{pmatrix}
                    \ a_{1,1}\ & a_{1,2}\ \\
                    \ a_{2,1}\ & a_{2,2}\ \\
                \end{pmatrix}
                \otimes
                \begin{pmatrix}
                    \ 1\ & 0\ & 0\ \\
                    \ 0\ & 1\ & 0\ \\
                    \ 0\ & 0\ & 1\ \\
                \end{pmatrix}
                +
                \begin{pmatrix}
                    \ 1\ & 0\ \\
                    \ 0\ & 1\ \\
                \end{pmatrix}
                \otimes
                \begin{pmatrix}
                         \ b_{1,1}\ & b_{1,2}\ & b_{1,3}\  \\
                         \ b_{2,1}\ & b_{2,2}\ & b_{2,3}\  \\
                         \ b_{3,1}\ & b_{3,2}\ & b_{3,3}\  \\
                \end{pmatrix}=\\
            %$$
            %\begin{eqnarray*}
                &&\begin{pmatrix}
                        %        1              2                 3                4                  5                  6
                        \ a_{1,1}          & 0\                & 0\               & a_{1,2}\         & 0\               & 0\         \\
                        \ 0\               & a_{1,1}           & 0\               & 0\               & a_{1,2}\         & 0\         \\
                        \ 0\               & 0\                & a_{1,1}\         & 0\               & 0\               & a_{1,2}\   \\
                        \ a_{2,1}\         & 0\                & 0\               & a_{2,2}\         & 0\               & 0\         \\
                        \ 0\               & a_{2,1}\          & 0\               & 0\               & a_{2,2}\         & 0\         \\
                        \ 0\               & 0\                & a_{2,1}\         & 0\               & 0\               & a_{2,2}\   \\
                \end{pmatrix}+
                \begin{pmatrix}
                        %        1              2                 3                4                  5                  6
                        \ b_{1,1}\         & b_{1,2}\          & b_{1,3}\         & 0\               & 0\               & 0\        \\
                        \ b_{2,1}\         & b_{2,2}\          & b_{2,3}\         & 0\               & 0\               & 0\        \\
                        \ b_{3,1}\         & b_{3,2}\          & b_{3,3}\         & 0\               & 0\               & 0\        \\
                        \ 0\               & 0\                & 0\               & b_{1,1}\         & b_{1,2}\         & b_{1,3}\  \\
                        \ 0\               & 0\                & 0\               & b_{2,1}\         & b_{2,2}\         & b_{2,3}\  \\
                        \ 0\               & 0\                & 0\               & b_{3,1}\         & b_{3,2}\         & b_{3,3}\  \\
                \end{pmatrix}=\\
            %$$
 %           {\scriptsize
 %           $$
                &&\begin{pmatrix}
                        %        1              2                 3                4                  5                  6
                        \ a_{1,1}+b_{1,1}\ & b_{1,2}\          & b_{1,3}\         & a_{1,2}\         & 0\               & 0\               \\
                        \ b_{2,1}\         & a_{1,1}+ b_{2,2}\ & b_{2,3}\         & 0\               & a_{1,2}\         & 0\               \\
                        \ b_{3,1}\         & b_{3,2}\          & a_{1,1}+b_{3,3}\ & 0\               & 0\               & a_{1,2}\         \\
                        \ a_{2,1}\         & 0\                & 0\               & a_{2,2}+b_{1,1}\ & b_{1,2}\         & b_{1,3}\         \\
                        \ 0\               & a_{2,1}\          & 0\               & b_{2,1}\         & a_{2,2}+b_{2,2}\ & b_{2,3}\         \\
                        \ 0\               & 0\                & a_{2,1}\         & b_{3,1}\         & b_{3,2}\         & a_{2,2}+b_{3,3}\ \\
                \end{pmatrix}.
 %           $$
 %           }
             \end{eqnarray*}
        \end{example}

        In the following we list basic properties of the Kronecker sum of matrices $A$, $B$, and $C$. Additional properties can be found in~\cite{PA:91} or are proved in this paper. The Kronecker sum is noncommutative because for element-wise comparison in general $A \oplus B \neq B \oplus A$. Anyway it essentially commutes because from a graph point of view, the graphs represented by matrices $A \oplus B$ and $B \oplus A$ are structurally isomorphic.

        Now we state a property of the Kronecker sum which we call \emph{Mixed Sum Rule}.
        \begin{lemma}{}\label{lemma_mixedSumRule}
        Let the matrices $A$ and $C$ have order $m$ and $B$ and $D$ have order $n$. Then we call $$(A \oplus B) + (C \oplus D) = (A + C) \oplus (B + D)$$ the \emph{Mixed Sum Rule}.
        %\footnote{Lemma~\ref{lemma_mixedSumRule} is used implicitly without proof in~\cite{BK:02}.}
        \end{lemma}
        \begin{proof}
            By using Eqs.~\eqref{equation_KroneckerProduct_bilinearity1}~and~\eqref{equation_KroneckerProduct_bilinearity2} and Def.~\ref{definition_kroneckerSum} we get $(A \oplus B) + (C \oplus D)= A \otimes I_n + I_m \otimes B + C \otimes I_n + I_m \otimes D=(A + C) \otimes I_n + I_m \otimes (B + D)=(A + C) \oplus (B + D)$.\qed
        \end{proof}

        For example let the matrices $A$ and $B$ be written as $A=\SumFromToText{i \in I}{}{A_i}$ and $B=\SumFromToText{j \in J}{}{B_j}$, respectively. In addition, let the sets $I$ and $J$ have the same number of elements, i.e., $|I| = |J|$. By using the mixed sum rule we can write $A \oplus B = \SumFromToText{i \in I, j \in J}{}{A_i \oplus B_j}$.

        We will frequently use the Mixed Sum Rule from now on without further notice.

        The Kronecker sum is also associative, as $(A \oplus B) \oplus C$ and
        $A \oplus (B \oplus C)$ are equal.

        \begin{lemma}{Kronecker sum is associative.
        %\footnote{Surprisingly we could not find a proof of Lemma~\ref{lemma_kroneckerSumIsAssociative}, although it is employed in literature many times.}
        }
        \label{lemma_kroneckerSumIsAssociative}
        \end{lemma}
        \begin{proof}
            In the following we will use $I_m\otimes I_n=I_{m.n}$. Note that $Z$ denotes zero matrices. We have
            { \setlength\arraycolsep{0.1em}
            \begin{eqnarray*}%missing spaces in comments!
                                                                                 A \oplus (B \oplus C)&=&A \oplus (B \otimes I_{o(C)}+I_{o(B)} \otimes C)\\
     \mbox{\{adding }Z_{o(A)}\}                                                                       &=&(A + Z_{o(A)}) \oplus (B \otimes I_{o(C)}+I_{o(B)} \otimes C)\\
     \mbox{\{Lemma~\ref{lemma_mixedSumRule}}\}                                                        &=&(A \oplus (B \otimes I_{o(C)}))+(Z_{o(A)} \oplus ( I_{o(B)} \otimes C))\\
     \mbox{\{Eq.\eqref{equation_kroneckerProduct_associativity}, Def.\ref{definition_kroneckerSum}}\} &=&(A \oplus (B \otimes I_{o(C)}))+I_{o(A)} \otimes I_{o(B)} \otimes C\\
     \mbox{\{ass.+, Def.\ref{definition_kroneckerSum}}\}                                              &=&A \otimes I_{o(B).o(C)}+I_{o(A)}\otimes B \otimes I_{o(C)}+\\
                                                                                                      & & I_{o(A).o(B)}\otimes C\\
     \mbox{\{comm. of }+\}                                                                            &=&A \otimes I_{o(B)} \otimes I_{o(C)}+I_{o(A).o(B)} \otimes C+\\
                                                                                                      & & I_{o(A)}\otimes B \otimes I_{o(C)}\\
     \mbox{\{Def.~\ref{definition_kroneckerSum}}\}                                                    &=&((A \otimes I_{o(B)}) \oplus C) + I_{o(A)} \otimes B \otimes I_{o(C)}\\
     \mbox{\{Def.~\ref{definition_kroneckerSum}}\}                                                    &=&((A \otimes I_{o(B)}) \oplus C) + (( I_{o(A)} \otimes B) \oplus Z_{o(C)}\\
     \mbox{\{Lemma~\ref{lemma_mixedSumRule}}\}                                                        &=&(A \otimes I_{o(B)} + I_{o(A)} \otimes B) \oplus (C + Z_{o(C)})\\
     \mbox{\{rm. }Z_{o(C)}\}                                                                          &=&(A \otimes I_{o(B)} + I_{o(A)} \otimes B) \oplus C\\
     \mbox{\{Def.~\ref{definition_kroneckerSum}}\}                                                    &=&(A \oplus B) \oplus C.
            \end{eqnarray*}
            }\qed
        \end{proof}

        \begin{comment}%full proof for associativity of Kronecker sum
                A \oplus (B \oplus C) &=& A \oplus (B \otimes I_{o(C)}+I_{o(B)} \otimes C)=\hfill\\
                && A \oplus (D+E)=(A+Z)\oplus(D+E)=\\
                && (A\oplus D)+(Z\oplus E)=\\
                && (A \oplus D) + (I_{o(A)} \otimes E)=\\
                && (A \oplus (B \otimes I_{o(C)}))+I_{o(A)} \otimes I_{o(B)} \otimes C=\\
                && A \otimes I_{o(D)}+I_{o(A)}\otimes D+I_{o(A)} \otimes I_{o(B)} \otimes C=\\
                && A \otimes I_{o(B).o(C)}+I_{o(A)}\otimes B \otimes I_{o(C)}+\\
                && I_{o(A).o(B)}\otimes C=\\
                && A \otimes I_{o(B)} \otimes I_{o(C)}+I_{o(A).o(B)} \otimes C+\\
                && I_{o(A)}\otimes B \otimes I_{o(C)}=\\
                && (A \otimes I_{o(B)} \oplus C) + I_{o(A)} \otimes B \otimes I_{o(C)}=\\
                && (F \oplus C)+ G \otimes I_{o(C)}=\\
                && (F \oplus C)+(G \oplus Z)=\\
                && (F+G)\oplus (C+Z)=(F+G)\oplus C\\
                && (A \otimes I_{o(B)} + I_{o(A)} \otimes B)\oplus C=\\
                (A \oplus B) \oplus C&&
        \end{comment}

        The associativity properties of the operations $\otimes$ and $\oplus$ imply that the k-fold operations $$\overset{k}{\underset{i=1}{\bigotimes}} A_i \ \ \ \hbox{and}\ \ \ \overset{k}{\underset{i=1}{\bigoplus}} A_i$$ are well defined.
        %folgender Block war drinnen
        %Let $n_i$ denote the order of matrix $A_i$ and $I_n$ the identity matrix of order $n$. Then we can write the n-fold Kronecker sum for $k$ matrices $A_i$, where $1 \leq i \leq k$ similar to Buchholz and Kemper~\cite{BK:02} and Ciardo et al.\,\cite{CM:99}

        % the following uses a product instead of exponent but the authors come from the same group as Buchholz
        % and Ciardo~\cite{CM:99}
        % and Plateau~\cite{}

        %folgender block war drinnen
        %as
        %$\overset{k}{\underset{i=1}{\bigoplus}} A_i = \SumFromTo{i=1}{k}{I_{n_1} \otimes \dots \otimes I_{n_{i-1}} %\otimes A_i \otimes I_{n_{i+1}} \otimes \dots \otimes I_{n_k}}=
        %\SumFromTo{i=1}{k}{I_{\overset{i-1}{\underset{j=1}{\prod}}n_j} \otimes A_i \otimes I_{\overset{k}{\underset{j=i+1}{\prod}}n_j}}.$

        %the short form on the right-hand side can be found wrong in literature e.g.~\cite{BCDK:97} uses $\SumFromTo{i=1}{k}{ I_{n^{i-1}_{1}} \otimes A^i \otimes I_{n^{k}_{i+1}} }$?}

        Note that Kronecker sum calculates all possible interleavings (see e.g.~\cite{Kue:91} for a proof). Note that this is true even for general CFGs including conditionals and loops. The following example illustrates interleaving of threads and how Kronecker sum handles it.

        \begin{example}\label{example_kronsumCD}
            \begin{figure}[t]
                \centering
                \subfigure[C]{
                    \includegraphics[scale=0.4]{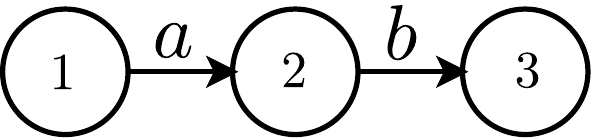}
                    \label{figure_C_interleavings}
                }
                \hspace{2mm}
                \subfigure[D]{
                    \includegraphics[scale=0.4]{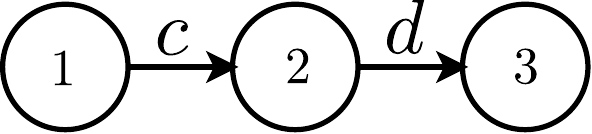}
                    \label{figure_D_interleavings}
                }
                \hspace{2mm} %\\
                \subfigure[Interleavings]{
                    %\begin{minipage}{3cm}
                    %\begin{center}
                    \begin{tabular}[b]{|c|}
                        \hline
                        \rule{0in}{3ex}\mbox{} Interleavings \mbox{}\rule{0in}{3ex} \\[0.5ex] %without mbox the label of the table is in the next line
                        \hline
                        $\rule{0in}{3ex}a \cdot b \cdot c \cdot d$\\
                        $a \cdot c \cdot b \cdot d$\\
                        $a \cdot c \cdot d \cdot b$\\
                        $c \cdot a \cdot b \cdot d$\\
                        $c \cdot a \cdot d \cdot b$\\
                        $c \cdot d \cdot a \cdot b$\\[0.1ex]
                        \hline
                    \end{tabular}
                    \label{table_interleavings}
                }
                \hspace{2mm}
                \subfigure[$C \oplus D$]{
                    \includegraphics[scale=0.4]{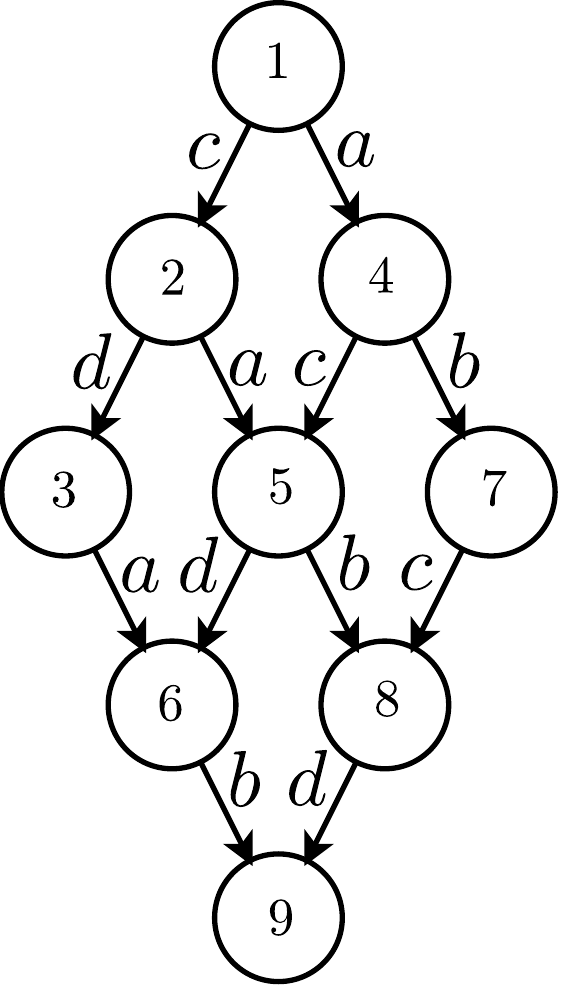}
                    \label{figure_CkronsumD}
                }
                \caption{A Simple Example}
                \label{figure_CDSystem}
            \end{figure}

            Let the matrices $C$ and $D$ be defined as follows: $$C=\begin{pmatrix}
                0 & a & 0\\
                0 & 0 & b\\
                0 & 0 & 0
            \end{pmatrix} \hskip5mm
            D=\begin{pmatrix}
                0 & c & 0\\
                0 & 0 & d\\
                0 & 0 & 0
            \end{pmatrix}.$$
            The graph corresponding to matrix $C$ is depicted in Fig.~\ref{figure_C_interleavings}, whereas the graph of matrix $D$ is shown in Fig.~\ref{figure_D_interleavings}.
            The regular expressions associated to the CFGs are $a \cdot b$ and $c \cdot d$, respectively. All possible interleavings by executing $C$ and $D$ in an interleavings semantics are shown in Fig.~\ref{table_interleavings}. In Fig.~\ref{figure_CkronsumD} the graph represented by the adjacency matrix $C \oplus D$ is depicted. It is easy to see that all possible interleavings are generated correctly.
        \end{example}

        \begin{comment}
        $C \oplus D =
        \begin{pmatrix}
         % 1  2   3   4   5   6   7   8   9
            & c & 0 & a & 0 & 0 & 0 & 0 & 0 \\ %1
            & 0 & d & 0 & a & 0 & 0 & 0 & 0 \\ %2
            & 0 & 0 & 0 & 0 & a & 0 & 0 & 0 \\ %3
            & 0 & 0 & 0 & c & 0 & b & 0 & 0 \\ %4
            & 0 & 0 & 0 & 0 & d & 0 & b & 0 \\ %5
            & 0 & 0 & 0 & 0 & 0 & 0 & 0 & b \\ %6
            & 0 & 0 & 0 & 0 & 0 & 0 & c & 0 \\ %7
            & 0 & 0 & 0 & 0 & 0 & 0 & 0 & d \\ %8
            & 0 & 0 & 0 & 0 & 0 & 0 & 0 & 0 \\ %9
        \end{pmatrix}$
        is depicted.
        \end{comment}
    \end{subsection}
\end{section} 

%% file: model.tex
\begin{section}{Concurrent Program Graphs}\label{section_CPG}
    Our system model consists of a finite number of threads and a finite number of semaphores. Both, threads and semaphores, are represented by CFGs. The CFGs are stored in form of adjacency matrices. The matrices have entries which are referred to as labels $l \in \cL$ as defined in Subsect.~\ref{subsection_overview}. Let $\cS$ and $\cT$ be the sets of adjacency matrices representing semaphores and threads, respectively. The matrices are manipulated by using Kronecker algebra. Similar to~\cite{BK:02} we describe synchronization by Kronecker products and thread interleavings by Kronecker sums. Note that higher synchronization features of programming languages such as Ada's rendezvous %or synchronization mechanisms of C\# and Java
    can be simulated by our system model as the runtime system
    %s of these languages
    uses semaphores provided by the operating systems to implement them.

    Formally, the system model consists of the tuple $\langle\cT, \cS, \cL \rangle$, where
    \begin{itemize}
        \item $\cT$ is the set of RCFG adjacency matrices describing threads,
        \item $\cS$ is the set of CFG adjacency matrices describing semaphores, and
        \item $\cL$ is the set of labels out of the semiring defined in Subsect.~\ref{subsection_overview}. The labels in $T \in \cT$ are elements of $\cL$, whereas the labels in $S \in \cS$ are elements of $\cL_{\SP}$.
    \end{itemize}

    A {\emph Concurrent Program Graph} (CPG) is a graph $C=\langle V,E,n_e \rangle$ with a set of nodes $V$, a set of directed edges $E~\subseteq~V~\times~V$, and a so-called \emph{entry} node $n_e \in V$. The sets $V$ and $E$ are constructed out of the elements of $\langle\cT, \cS, \cL \rangle$. Details on how we generate the sets $V$ and $E$ follow in the next subsections. Similar to RCFGs the edges of CPGs are labeled by $l \in \cL$.
    Assuming without loss of generality that each thread has an entry node with index $1$ in its adjacency matrix $t \in \cT$, then the entry node of the generated CPG has index $1$, too.

    In Fig.~\ref{figure_overview} an overview of our approach is given. As described in Subsect.~\ref{subsection_edgeSplitting} the set of shared variables $\cV$ is used to generate $\cT$.

    \begin{figure}[t]
        \centering
        \includegraphics[scale=0.6]{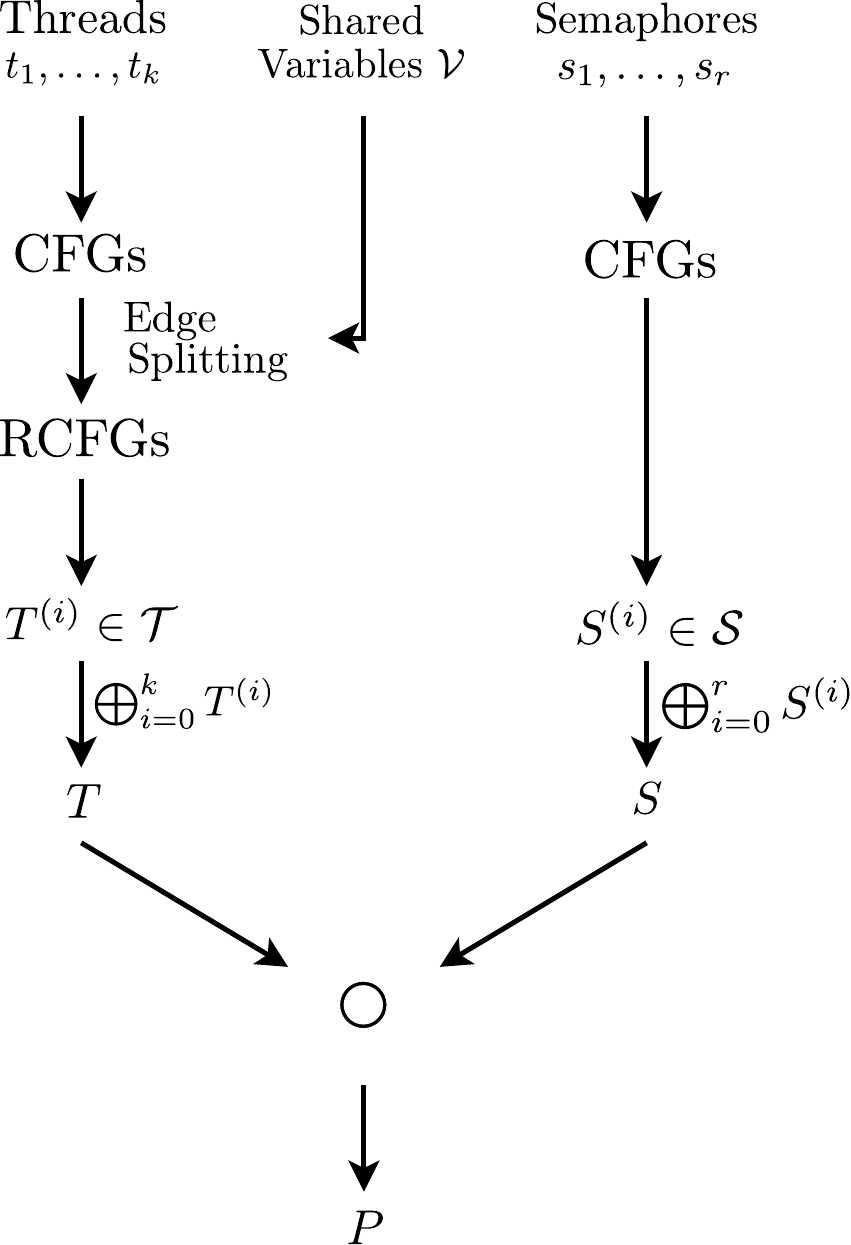}
        \caption{Overview}
        \label{figure_overview}
    \end{figure}

    \begin{subsection}{Generating a Concurrent Program's Matrix}
        Let $T^{(i)} \in \cT$ and $S^{(i)} \in \cS$ refer to the matrices representing thread $i$ and semaphore $i$, respectively. Let $M=(m_{i,j}) \in \cM$. In addition, we define the matrix $M_l$ as the matrix with entries of $M$ equal to $l$ and zeros elsewhere:
        $$M_l=(m_{l;i,j})\mbox{, where } m_{l;i,j}=
                \left \{ \begin{array}{ll}
                            l & \mbox{if }m_{i,j}=l,\\
                            0 & \mbox{otherwise.}
                        \end{array}
                \right.
        $$

        We obtain the matrix representing the $k$ interleaved threads as
        $$T=\overset{k}{\underset{i=1}{\bigoplus}} \ T^{(i)}\mbox{, where } T^{(i)} \in \cT.$$

        According to Fig.~\ref{figure_semaphores} we have for the binary and the counting semaphore an adjacency matrix of order two and three, respectively. If we assume that the $i$th and the $j$th semaphore, where $1 \leq i,j \leq r$, are a  binary and a counting semaphore, respectively, then we get the following adjacency matrices.
        $$  S^{(i)}=
            \begin{pmatrix}
                0   & p_i\\
                v_i & 0\\
            \end{pmatrix}
            \hbox{ and } S^{(j)}=
            \begin{pmatrix}
                0   & p_j & 0\\
                v_j & 0   & p_j\\
                0   & v_j & 0\\
            \end{pmatrix}
        $$
        In a similar fashion we can model counting semaphores of higher order.

        The matrix representing the $r$ interleaved semaphores is given by
        $$S=\overset{r}{\underset{i=1}{\bigoplus}} \ S^{(i)}\mbox{, where } S^{(i)} \in \cS.$$

%        With the above prerequisites we are now able to give the formula for of a concurrent programs operator $\merry$ (similar to~\cite{BK:02}):

        %\begin{eqnarray}
        %    M \merry N
        %& = & \SumFromTo{l \in \cL_{\SP_p}}{}{\left ( M_{l} \otimes N_{l} \right )}+ %sync points p
        %      \SumFromTo{l \in \cL_{\SP_v}}{}{\left ( M_{l} \otimes N_{l} \right )}+\\ %sync points v
        %&   & \SumFromTo{l \in \cL_{\V}}{}{\left ( M_{l} \oplus N_{l} \right )}.
        %\end{eqnarray}

        The adjacency matrix representing program $\cP$ referred to as $P$ is defined as
        \begin{equation}\label{definition_merryOperator}
            P=T \merry S =
              %\SumFromTo{l \in \cL_{\SP_p}}{}{\left ( M_{l} \otimes N_{l} \right )}+ %sync points p
              %\SumFromTo{l \in \cL_{\SP_v}}{}{\left ( M_{l} \otimes N_{l} \right )}+\\ %sync points v
              \SumFromTo{l \in \cL_{\SP}}{}{\left ( T_{l} \otimes S_{l} \right )}+\\ %sync points
              \SumFromTo{l \in \cL_{\V}}{}{\left ( T_{l} \oplus S_{l} \right )}.
        \end{equation}

        When applying the Kronecker product to semaphore calls we follow the rules $v_x \cdot v_x=v_x$ and $p_x \cdot p_x=p_x$.

        In Subsect.~\ref{subsection_efficientImplOfMerryOperation} we describe how the $\merry$-operation can be implemented efficiently.

        %\begin{eqnarray*}%n-fold version of \merry-Operator
        %       \overset{j}{\underset{i=1}{\merryBig}} M^{(i)}
        % & = & \SumFromTo{l \in \cL_{\SP_p}}{}{\left ( \overset{j}{\underset{i=1}{\bigotimes}} M^{(i)}_{l}} \right )+ %sync points p
        %       \SumFromTo{l \in \cL_{\SP_v}}{}{\left ( \overset{j}{\underset{i=1}{\bigotimes}} M^{(i)}_{l}} \right )+\\ %sync points v
        % &   & \SumFromTo{l \in \cL_{\V}}{}{\left ( \overset{j}{\underset{i=1}{\bigoplus}} M^{(i)}_{l}} \right ).
        %\end{eqnarray*}
    \end{subsection}

    \begin{subsection}{$\merry$-Operation and Synchronization}\label{subsection_proofSyncInOurModel}
        \begin{lemma}\label{semsync}
            Let $T=\bigoplus_{i=1}^k T^{(i)}$ be the matrix representing $k$ interleaved threads and let $S$ be a binary semaphore.
            Then $T \merry S$ correctly models synchronization of $T$ with semaphore $S$.\footnote{Note that we do not make assumptions concerning the structure of $T$.}
        \end{lemma}
        \begin{proof}
            First we observe that
            \begin{enumerate}
            \item
            the first term in the definition of Eq.~\eqref{definition_merryOperator} replaces
            \begin{itemize}
                \item each $p$ in matrix $T$ with $\begin{pmatrix}0&p\\0&0\end{pmatrix}$ and
                \item each $v$ in matrix $T$ with $\begin{pmatrix}0&0\\v&0\end{pmatrix}$,
            \end{itemize}
            \item
            the second term replaces each $m \in \cL_{\cV}$ with $\begin{pmatrix}m&0\\0&m\end{pmatrix}$, and
            \item
            both terms replace each $0$ by $\begin{pmatrix}0 & 0 \\ 0 & 0\end{pmatrix}$.
            \end{enumerate}
            According to the replacements above the order of matrix $T \merry S$ has doubled compared to $T$.

            Now, consider the paths in the automaton underlying $T$ described by the regular expression
            $$\pi=\left( \SumFromToText{m\in \cL_{\cV}}{}{m} \right)^{*} \left (p \left(  \SumFromToText{m\in \cL_{\cV}}{}{m}  \right)^{*} v \left ( \SumFromToText{m\in \cL_{\cV}}{}{m} \right)^{*} \right )^{*}.$$

            By the observations above it is easy to see that paths containing $\pi$ are present in $T \merry S$.
            On the other hand, paths not containing $\pi$ are no more present in $T \merry S$.
            Thus the semaphore operations always occur in $(p,v)$ pairs in all paths in $T \merry S$.
            This, however, exactly mirrors the semantics of synchronization via a semaphore. \qed
        \end{proof}

        Generalizing Lemma~\ref{semsync}, it is easy to see that the synchronization property is also correctly modeled if we replace the binary semaphore by one which allows more than one thread to enter it.
        In addition, the synchronization property is correctly modeled even if more than one semaphore is present on the right-hand side of $T \merry S$.

        As a byproduct the proof of Lemma~\ref{semsync} shows the following corollary.
        \begin{corollary}\label{corollary_zeroline}
            If the program modeled by $T \merry S$ contains a deadlock, then the matrix $T \merry S$ will contain a zero line $\ell$.
            Node $\ell$ in the corresponding automaton is no final node and does not have successors.
        \end{corollary}
        Thus deadlocks show up in CPGs as a pure structural property of the underlying graphs. Nevertheless, false positives may occur. From a static point of view, a deadlock is possible while conditions exclude this case at runtime. Our approach delivers a path to a deadlock in any case. Nevertheless, our approach of finding deadlocks is complete. If it states deadlock freedom, then the program under test is certainly deadlock free.

        A further consequence of Lemma~\ref{semsync} is that after applying the $\merry$-operation only a small part of the underlying automata can be reached from its entry node. This allows for optimizations discussed later.
    \end{subsection}

    \begin{subsection}{Unreachable Parts Caused by Synchronization}\label{subsection_unreachableExample}
        In this subsection we show that synchronization causes unreachable parts. As an example consider Fig.~\ref{figure_exampleMutualExclusion}.
        The program consists of two threads, namely $T_1$ and $T_2$. The RCFGs of the threads are shown in Fig.~\ref{figure_exampleMutualExclusionT1} and Fig.~\ref{figure_exampleMutualExclusionT2}. The used semaphore is a binary semaphore similar to Fig.~\ref{figure_semaphoreBinary}. Its operations are referred to as $p_1$ and $v_1$. We denote a P and V-call to semaphore $x$ of thread $t$ as $t.p_x$ and $t.v_x$, respectively. $T_1$ and $T_2$ access the same shared variable in $a$ and $b$, respectively. The semaphore is used to ensure that $a$ and $b$ are accessed mutually exclusively. Note that $a$ and $b$ may actually be subgraphs consisting of multiple nodes and edges.

        For the example we have the matrices
        $$
        T_1=\begin{pmatrix}
                0  & p_1 & 0 & 0\\
                0  & 0  & a & 0\\
                0  & 0  & 0 & v_1\\
                0  & 0  & 0 & 0\\
            \end{pmatrix}
        \mbox{, }
        T_2=\begin{pmatrix}
                0  & p_1 & 0 & 0\\
                0  & 0  & b & 0\\
                0  & 0  & 0 & v_1\\
                0  & 0  & 0 & 0\\
            \end{pmatrix}
        \mbox{, and } S=\begin{pmatrix}
                0  & p_1\\
                v_1 & 0\\
            \end{pmatrix}.
        $$
        %In addition let $N=\begin{pmatrix}
        %        0 & p_1 & 0 & 0\\
        %        0 & 0 & b & 0\\
        %        0 & 0 & 0 & v_1\\
        %        0 & 0 & 0 & 0\\
        %    \end{pmatrix}.$

        \noindent
        Then we obtain the matrix $T=T_1\oplus T_2$, a matrix of order 16, consisting of the submatrices
        defined above and zero matrices of order four (instead of $Z_4$ simply denoted by 0) as follows.
        $$T=\begin{pmatrix}
                \ T_2\ & p_1 \cdot I_4 & 0           & 0            \\
                0      & T_2           & a \cdot I_4 & 0            \\
                0      & 0             & T_2         & v_1 \cdot I_4\\
                0      & 0             & 0           & T_2          \\
            \end{pmatrix}
        $$
        In order to enable a concise presentation of $T \circ S$ we define the matrices
        \begin{eqnarray*}
        U&=&\begin{pmatrix}
                0 & 0 & 0 & p_1 & 0 & 0 & 0 & 0\\
                0 & 0 & 0 & 0 & 0 & 0 & 0 & 0\\
                0 & 0 & 0 & 0 & b & 0 & 0 & 0\\
                0 & 0 & 0 & 0 & 0 & b & 0 & 0\\
                0 & 0 & 0 & 0 & 0 & 0 & 0 & 0\\
                0 & 0 & 0 & 0 & 0 & 0 & v_1 & 0\\
                0 & 0 & 0 & 0 & 0 & 0 & 0 & 0\\
                0 & 0 & 0 & 0 & 0 & 0 & 0 & 0\\
            \end{pmatrix},
          V=\begin{pmatrix}
                0 & p_1 & 0 & 0 & 0 & 0 & 0 & 0\\
                0 & 0 & 0 & 0 & 0 & 0 & 0 & 0\\
                0 & 0 & 0 & p_1 & 0 & 0 & 0 & 0\\
                0 & 0 & 0 & 0 & 0 & 0 & 0 & 0\\
                0 & 0 & 0 & 0 & 0 & p_1 & 0 & 0\\
                0 & 0 & 0 & 0 & 0 & 0 & 0 & 0\\
                0 & 0 & 0 & 0 & 0 & 0 & 0 & p_1\\
                0 & 0 & 0 & 0 & 0 & 0 & 0 & 0\\
            \end{pmatrix},\\
          W&=&a \cdot I_8\mbox{, and }
          X=\begin{pmatrix}
                0 & 0 & 0 & 0 & 0 & 0 & 0 & 0\\
                v_1 & 0 & 0 & 0 & 0 & 0 & 0 & 0\\
                0 & 0 & 0 & 0 & 0 & 0 & 0 & 0\\
                0 & 0 & v_1 & 0 & 0 & 0 & 0 & 0\\
                0 & 0 & 0 & 0 & 0 & 0 & 0 & 0\\
                0 & 0 & 0 & 0 & v_1 & 0 & 0 & 0\\
                0 & 0 & 0 & 0 & 0 & 0 & 0 & 0\\
                0 & 0 & 0 & 0 & 0 & 0 & v_1 & 0\\
            \end{pmatrix}\mbox{ of order 8.}
        \end{eqnarray*}
        Then we obtain the matrix $T \circ S$, a matrix of order 32, consisting of the submatrices
        defined above and zero matrices of order eight (instead of $Z_8$ simply denoted by 0) as follows.

        $$T \circ S=
            \begin{pmatrix}
                U & V & 0 & 0 \\
                0 & U & W & 0 \\
                0 & 0 & U & X \\
                0 & 0 & 0 & U \\
            \end{pmatrix}.
        $$

        The generated CPG is depicted in Fig.~\ref{figure_exampleMutualExclusionCPG}. The resulting adjacency matrix has order 32, whereas the resulting CPG consists only of 12 nodes and 12 edges. Large parts (20 nodes and 20 edges) are unreachable from the entry node. In Fig.~\ref{figure_exampleMutualExclusionUnreachableParts} these unreachable parts are depicted.

        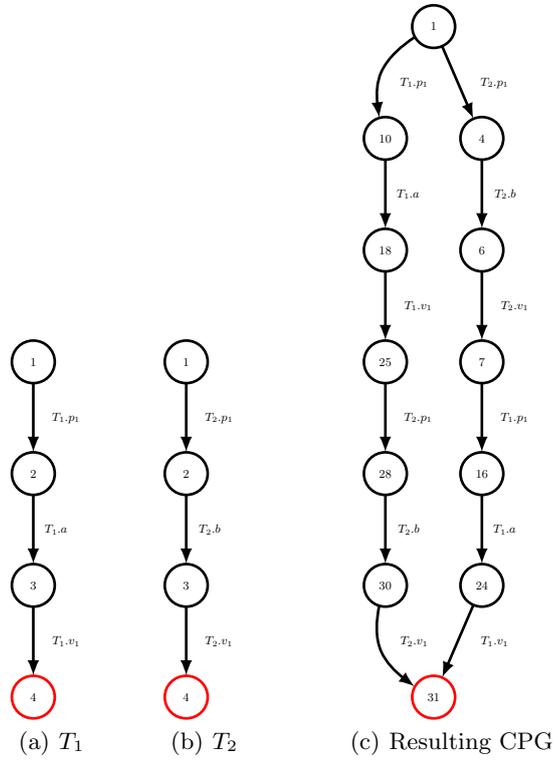
\begin{figure}%[tbh]
            \centering
            \subfigure[$T_1$]{
                \input{./dot2tex_figs/MutualExclusionOverOneSemaphoreSmall_T1.tex}
                \label{figure_exampleMutualExclusionT1}
            }
            \hspace{5mm}
            \subfigure[$T_2$]{
                %\hspace{5mm}
                \input{./dot2tex_figs/MutualExclusionOverOneSemaphoreSmall_T2.tex}
                %\hspace{5mm}
                \label{figure_exampleMutualExclusionT2}
            }
            \hspace{5mm}
            \subfigure[Resulting CPG]{
                \hspace{5mm}
                \input{./dot2tex_figs/MutualExclusionOverOneSemaphoreSmall_CPG.tex}
                \hspace{5mm}
                \label{figure_exampleMutualExclusionCPG}
            }
            \caption{Mutual Exclusion Example}
            \label{figure_exampleMutualExclusion}
        \end{figure}
        \begin{figure}%[tbh]
            \centering
            \input{./dot2tex_figs/MutualExclusionOverOneSemaphoreSmallUnreachable.tex}
            \caption{Unreachable Parts of the Mutual Exclusion Example}
            \label{figure_exampleMutualExclusionUnreachableParts}
        \end{figure}
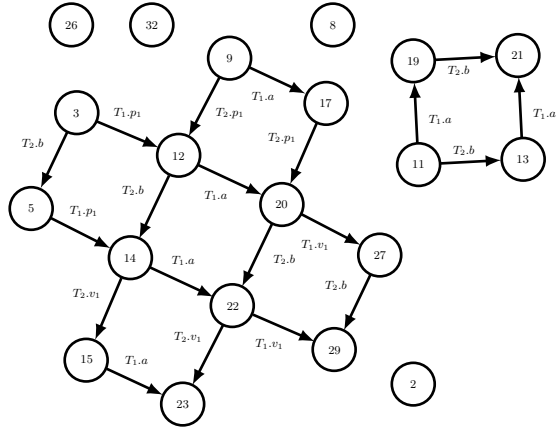
%
        \begin{comment}
        \begin{figure}%[tbh]
            \centering
            \subfigure[$T_1$]{
                \includegraphics[scale=0.5]{./dot2tex_figs/MutualExclusionOverOneSemaphoreSmall_T1_afterInkscape}
                \label{figure_exampleMutualExclusionT1}
            }
            \hspace{5mm}
            \subfigure[$T_2$]{
                %\hspace{5mm}
                \includegraphics[scale=0.5]{./dot2tex_figs/MutualExclusionOverOneSemaphoreSmall_T2_afterInkscape}
                %\hspace{5mm}
                \label{figure_exampleMutualExclusionT2}
            }
            \hspace{5mm}
            \subfigure[Resulting CPG]{
                \hspace{5mm}
                \includegraphics[scale=0.4]{./dot2tex_figs/MutualExclusionOverOneSemaphoreSmall_CPG_afterInkscape}
                \hspace{5mm}
                \label{figure_exampleMutualExclusionCPG}
            }
            \caption{Mutual Exclusion Example}
            \label{figure_exampleMutualExclusion}
        \end{figure}

        \begin{figure}%[tbh]
            \centering
            \includegraphics[scale=0.5]{./dot2tex_figs/MutualExclusionOverOneSemaphoreSmallUnreachable_afterInkscape}
            \caption{Unreachable Parts of the Mutual Exclusion Example}
            \label{figure_exampleMutualExclusionUnreachableParts}
        \end{figure}
        \end{comment}

        In general, unreachable parts exist if a concurrent program contains synchronization. If a program contains a lot of synchronization the reachable parts may be very small. This observation motivates the lazy implementation described in Subsect.~\ref{subsection_lazyImpl}.
    \end{subsection}

    \begin{subsection}{Properties of the Resulting Adjacency Matrix}%\mbox{}\\
        In this subsection we prove interesting properties of the resulting matrices.

        %In case matrix $A$ has order m and $B$ has order $n$ it is easy to justify that $A \oplus B$ has order $m n$ (cf. Def.~\ref{definition_kroneckerSum}).
        A short calculation shows that the Kronecker sum in general generates at most $m n^2 + n m^2 - nm$ non-zero entries.\footnote{Assuming the corresponding matrices have an order of $m$ and $n$, respectively.} Stated the other way, at least $(m n)^2 - m n^2 - n m^2 + m n$ entries are zero. We will see that CFGs and RCFGs contain even more zero entries. We will prove that for this case the number of edges is in $\bigo{mn}$. Thus, the number of edges is linear in the order of the resulting adjacency matrix.

        %Given a program $\cP$ consisting of $k>0$ threads $(t_1, t_2, \dots , t_k)$, where each $t_i$ has $n$ nodes in its RCFG, the maximum order of $P$'s adjacency matrix is given by $n^k$. This property is given by the definition of the Kronecker sum.

        \begin{lemma}[Maximum Number of Nodes]\label{lemma_maxNumOfNodesInCPG}
            Given a program $\cP$ consisting of $k>0$ threads $(t_1, t_2, \dots , t_k)$, where each $t_i$ has $n$ nodes in its RCFG, the number of nodes in $\cP$'s adjacency matrix $P$ is bounded from above by $n^k$.
        \end{lemma}
        \begin{proof}This follows immediately from the definitions of $\otimes$ and $\oplus$. For both the order of the resulting matrix is given by the multiplication of the orders of the input matrices.\qed
        \end{proof}

        \begin{definition}
            Let $M=(m_{i,j}) \in \cM$. We denote the number of non-zero entries by $||M||=|\{m_{i,j} \, | \, m_{i,j} \neq 0\}|$.
        \end{definition}

        For a RCFG with $n$ nodes it is easy to see that it contains at most $2n$ edges.
        %For sake of simplicity we use $2n$ as an upper bound.

        \begin{lemma}[Maximum Number of Entries]\label{lemma_maxNumOfEdgesInCPGS}
            %Given a program $\cP$ consisting of $k>0$ threads $(t_1, t_2, \dots , t_k)$, where each $t_i$ has $n$ nodes in its RCFG, the maximum number of entries in $\cP$'s adjacency matrix $P$ is bounded from above by $2 k \, n^k$.
            Let a program represented by $M_k \in \cM$ consisting of $k>0$ threads be represented by the matrices $T^{(i)} \in \cT$, where each $T^{(i)}$ has order $n$. Then $||M_k||$ is bounded from above by $2 k \, n^k$.
        \end{lemma}
        \begin{proof} We prove this lemma by induction on the definition of the Kronecker sum.
            For $k=1$ the lemma is true. If we assume that for $m$ threads $||M_m|| \leq 2m \, n^m$, then for $m+1$ threads $||M_{m+1}|| \leq 2 m \, n^m \cdot n + n^m \cdot 2 n = 2 (m+1) \, n^{m+1}$. Thus, we have proved Lemma~\ref{lemma_maxNumOfEdgesInCPGS}.\qed
            %$k=1$: $2n$ (this is the assumption for each thread),\\
            %$k=2$: $2n \cdot n + n \cdot 2n = 4n^2$,\\
            %$k=3$: $4n^2*n+n^2 \cdot 2n = 6n^3$,\\
            %$m$ threads: $2m \, n^m$,\\
            %$m+1$ threads: $2 m \, n^m \cdot n + n^m \cdot 2 n = 2 (m+1) \, n^{m+1}$.\\
        \end{proof}

        Compared to the full matrix of order $n^k$ with $n^{2k}$ entries the resulting matrix has significantly fewer non-zero entries, namely $2 k \, n^k$. %For example the maximum number of edges in $A \oplus B$ is given by $4nm$.
        By using the following definition we will prove that the matrices are sparse.

        \begin{definition}[Sparse Matrix]\label{definition_SparseMatrix}
            We call a n-by-n matrix $M$ sparse if and only if $||M||=\bigo{n}$.
        \end{definition}

        \begin{lemma}\label{lemma_CFGsAreSparseGraphs}
            CFGs and RCFGs have Sparse Adjacency Matrices.
        \end{lemma}

        \begin{proof}
              %Let $G = \langle V, E, n_e \rangle$ be a CFG consisting of $|V|$ nodes. From
              %Subsect.~\ref{subsection_controlFlowGraphs}
              %we know that the number of edges $|E|$ in $G$ is bounded from above by $2 \, |V|$. By using Definition~\ref{definition_SparseMatrix} we get for the maximum number of %entries in the adjacency matrix $M$:
              %$||M|| = 2 \, |V| = \bigo{|V|}$.
              Follows from Subsect.~\ref{subsection_controlFlowGraphs} and Def.~\ref{definition_SparseMatrix}.\qed
        \end{proof}

        \begin{lemma}\label{lemma_CPGsAreSparseGraphs}
            The Matrix $P$ of a Program $\cP$ is Sparse.
        \end{lemma}

        \begin{proof}
              Let $T=\bigoplus_{i=1}^{k}{T^{(i)}}\in \cM$ be a N-by-N adjacency matrix of a program. We require that each of the $k$ threads has order $n$ in its adjacency matrix $T^{(i)}$. From Lemma~\ref{lemma_maxNumOfEdgesInCPGS} we know $||T||=\bigo{2 k \, n^k}$. In addition, $N=n^k$ is given by Lemma~\ref{lemma_maxNumOfNodesInCPG}. Hence, for $k$ threads and by using Definition~\ref{definition_SparseMatrix} we get $||T|| \leq 2 k \, n^k=2 k \, N=\bigo{N}$. A similar result holds for S and $P = T \merry S$.\qed
        \end{proof}

        Lemma~\ref{lemma_CPGsAreSparseGraphs} enables the application of memory saving data structures and efficient algorithms. Algorithms may for example work on adjacency lists. Clearly, the space requirements for the adjacency lists are linear in the number of nodes. In the worst-case the number of nodes increases exponentially in the number of threads.
    \end{subsection}

    \begin{subsection}{Efficient Implementation of the $\merry$-Operation}\label{subsection_efficientImplOfMerryOperation}
        This subsection is devoted to an efficient implementation of the $\merry$-operation.
        First we define the Selective Kronecker product which we denote by $\oslash$. This operator synchronizes only identical labels $l \in \cL_{\SP}$ of the two input matrices.
        \begin{definition}[Selective Kronecker product]
            Given two matrices $A$ and $B$ we call $A \oslash_L B$ their Selective Kronecker product. For all $l \in L \subseteq \cL$ let $A \oslash_L B = (a_{i,j}) \oslash_L (b_{p,q}) = (c_{i.p,j.q})$, where
            $$c_{i.p,j.q} =
                \left \{ \begin{array}{ll}
                    l & \mbox{\qquad if }a_{i,j}=b_{p,q}=l \wedge l \in L,\\
                    0 & \mbox{\qquad otherwise.}
                \end{array}
                \right.
            $$
        \end{definition}

        \begin{definition}[Filtered Matrix]
            We call $M_L$ a \emph{Filtered Matrix} and define it as a matrix of order $o(M)$ containing entries $l \in L \subseteq \cL$ of $M$ and zeros elsewhere as follows.
            $$M_L=(m_{L;i,j})\mbox{, where } m_{L;i,j}=
                \left \{
                \begin{array}{ll}
                    l & \mbox{\qquad if }m_{i,j}=l \wedge l \in L,\\
                    0 & \mbox{\qquad otherwise.}
                \end{array}
                \right.
            $$
        \end{definition}

        Note that
        \begin{equation}\label{equation_equalityOfSelectiveKroneckerProduct}
            \SumFromTo{l \in \cL_{\SP}}{}{\left ( T_{l} \otimes S_{l} \right )}=T \oslash_{\cL_{\SP}} S.
        \end{equation}

        In the following we use $o(S_{\cL_{\V}})=\ProdFromToText{i=1}{r}{o(S^{(i)})}=o(S)$. Note that S contains only labels $l \in \cL_{\SP}$. Hence, when the $\merry$-operator is applied for a label $l \in \cL_{\V}$, we get $S_l=Z_{o(S)}$, i.e. a zero matrix of order $o(S)$. Thus we obtain $\SumFromToText{l \in \cL_{\V}}{}{\left ( T_{l} \oplus S_{l} \right )}=T_{\cL_{\V}} \otimes I_{o(S)}$. We will prove this below.

        Finally, we can refine Eq.~\eqref{definition_merryOperator} by stating the following lemma.

        \begin{lemma}{}\label{lemma_merryOperatorCanBeComputedEfficiently}
            The $\merry$-operation can be computed efficiently by\\
            $$P=T \merry S=T \oslash_{\cL_{\SP}} S+T_{\cL_{\V}} \otimes I_{o(S)}.$$
        \end{lemma}
        \begin{proof}
            Using Eq.~\eqref{definition_merryOperator} $P=T \merry S$ is given by
            $\SumFromTo{l \in \cL_{\SP}}{}{\left ( T_{l} \otimes S_{l} \right )}+ %sync points
                        \SumFromTo{l \in \cL_{\V}}{}{\left ( T_{l} \oplus S_{l} \right )}.$
            According to Eq.~\eqref{equation_equalityOfSelectiveKroneckerProduct} the first term is equal to $T \oslash_{\cL_{\SP}} S$.
            By mentioning $S_{l}=Z_{o(S)}$ for $l\in \cL_{\V}$, Lemma~\ref{lemma_mixedSumRule}, and Def.~\ref{definition_kroneckerSum}, the second term fulfills.
            $$\SumFromTo{l \in \cL_{\V}}{}{\left ( T_{l} \oplus S_{l} \right )}= \SumFromTo{l \in \cL_{\V}}{}{\left ( T_{l} \oplus Z_{o(S)} \right )}=
            T_{\cL_{\V}} \oplus Z_{o(S)}=T_{\cL_{\V}} \otimes I_{o(S)}.$$
            Note that $S$ contains only $l \in \cL_{\SP}$. It is obvious that the non-zero entries of the first and the second term are $l \in \cL_{\SP}$ and $l \in \cL_{\V}$, respectively. Both terms can be computed by iterating once through the corresponding sparse adjacency matrices, namely $T$ and $S$.\qed
        \end{proof}
    \end{subsection}

    \begin{subsection}{Lazy Implementation of Kronecker Algebra}\label{subsection_lazyImpl}
        Until now we have primarily focused on a pure mathematical model for shared memory concurrent systems.
        An alert reader will have noticed that the order of the matrices in our CPG increases exponentially in the number of threads.
        On the other hand, we have seen that the $\merry$-operation results in parts of the matrix $T \merry S$ that cannot be reached from the entry node of the underlying automaton (cf. Subsect.~\ref{subsection_unreachableExample}). This comes solely from the fact that synchronization excludes some interleavings.

        Choosing a lazy implementation for the matrix operations, however, ensures that,
        when extracting the reachable parts of the underlying automaton, the overall effort is reduced to exactly these parts.
        By starting from the entry node and calculating all reachable successor nodes our lazy implementation exactly does this.
        Thus, for example, if the resulting automaton's size is linear in terms of the involved threads, only linear effort will be necessary to generate the resulting automaton.

        Our implementation distinguishes between two kind of matrices: Sparse matrices are used for representing threads and semaphores.
        Lazy matrices are employed for representing all the other matrices, e.g. those resulting from the operations of the Kronecker algebra and our $\merry$-operation.
        Besides the employed operation, a lazy matrix simply keeps track of its operands.
        Whenever an entry of a lazy matrix is retrieved, depending on the operation recorded in the lazy matrix, entries of the operands are retrieved and the recorded operation is performed on these
        entries to calculate the result.
        In the course of this computation, even the successors of nodes are evaluated lazily.
        Retrieving entries of operands is done recursively if the operands are again lazy matrices, or is done by retrieving the entries from the sparse matrices, where the actual data resides.

        In addition, our lazy implementation allows for simple parallelizing.
        For example, retrieving the entries of left and right operands can be done concurrently.
        Exploiting this, we expect further performance improvements for our implementation if run on multi-core architectures.
    \end{subsection}

    \begin{subsection}{Optimization for NSV}\label{subsection_optimizationForNSV}
        Our approach already works fine for practical settings. In this subsection we present additional optimizations which are optional.

        As already mentioned in Subsect.~\ref{subsection_KroneckerAlgebraIntroduction} the Kronecker sum interleaves all entries.
        Sometimes this is disadvantageous because irrelevant interleavings will be generated if some basic blocks do not access shared variables.
        Such basic blocks can be placed freely as long as other constraints do not prohibit it.

        For example consider the CFGs in Fig.~\ref{figure_CDSystem}.
        Assume for a moment that $a$, $b$, $c$, and $d$ do not access shared variables.
        Then the overall behavior of the $C$-$D$-system can be described correctly by choosing {\em one} of the six interleavings depicted in Fig.~\ref{figure_CkronsumD}, e.g., by $a \cdot b \cdot c \cdot d$.
        Hence the size of the CPG is reduced from nine nodes to five.

        From now on we divide set ${\cL}_{\V}$ into two disjoint sets ${\cL}_{\SV}$ and ${\cL}_{\NSV}$ depending on whether the corresponding basic blocks access shared variables or not.

        The following example shows that NSV-edges cannot always be eliminated.

        \begin{example}\label{example_oneNSVEdgeIsNotSufficient}
            In this example we use the graphs depicted in Fig.~\ref{figure_EFSystem}. The graphs $E$ and $F$ form the input graphs. It is assumed that $a$ is the only edge not accessing a shared variable. All graphs have Node $1$ as entry node. We show that it is not sufficient to chose exactly one NSV-edge.
            The matrix $E\oplus F$ is given by \\
            $$ { %\scriptsize
            \begin{pmatrix}
                % 1   2   3   4   5   6   7   8
                  0 & c & a & 0 & 0 & 0 & 0 & 0 \\ %1
                  0 & 0 & 0 & a & 0 & 0 & 0 & 0 \\ %2
                  0 & 0 & 0 & c & p & 0 & 0 & 0 \\ %3
                  0 & 0 & 0 & 0 & 0 & p & 0 & 0 \\ %4
                  0 & 0 & 0 & 0 & 0 & c & b & 0 \\ %5
                  0 & 0 & 0 & 0 & 0 & 0 & 0 & b \\ %6
                  v & 0 & 0 & 0 & 0 & 0 & 0 & c \\ %7
                  0 & v & 0 & 0 & 0 & 0 & 0 & 0 \\ %8
            \end{pmatrix}.
            }
            $$

            The graph represented by $E \oplus F$ which is structurally isomorph to $(E \oplus F) \merry S$ is depicted in Fig.\ref{figure_EF}. Both loops in the CPG must be preserved. Otherwise the program would be modeled incorrectly. By removing an edge labeled by $a$, we would change the program behavior. Thus it is not sufficient to use only one edge labeled by $a$.

            \begin{figure}[t]
                \centering
                \subfigure[E]{
                    \includegraphics[scale=0.5]{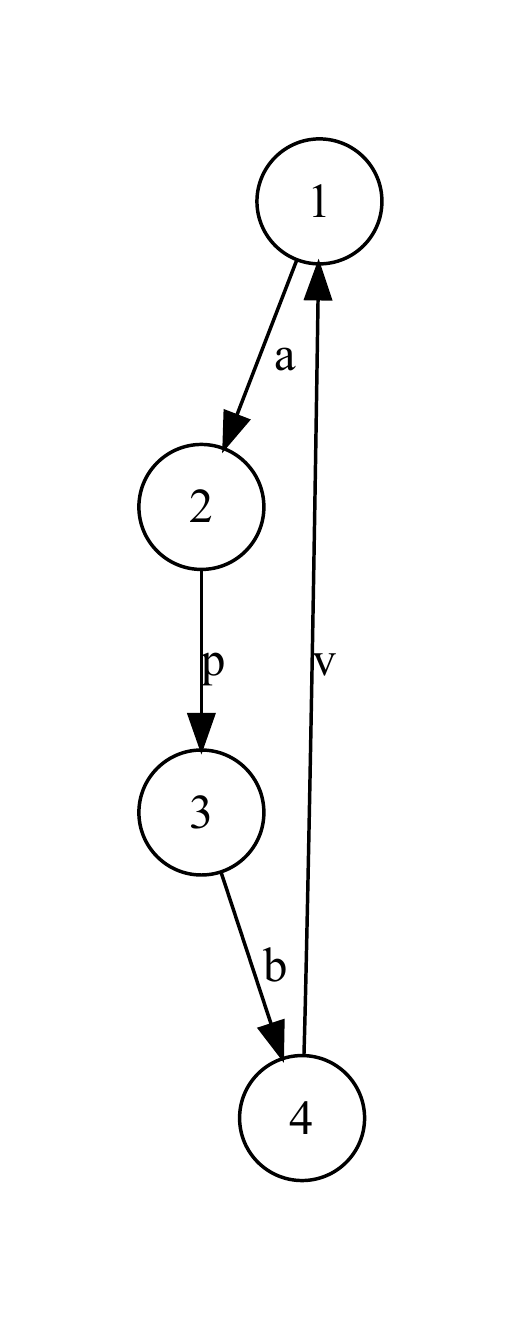}
                    \label{figure_E}
                }
                \hspace{8mm}
                \subfigure[F]{
                    \includegraphics[scale=0.5]{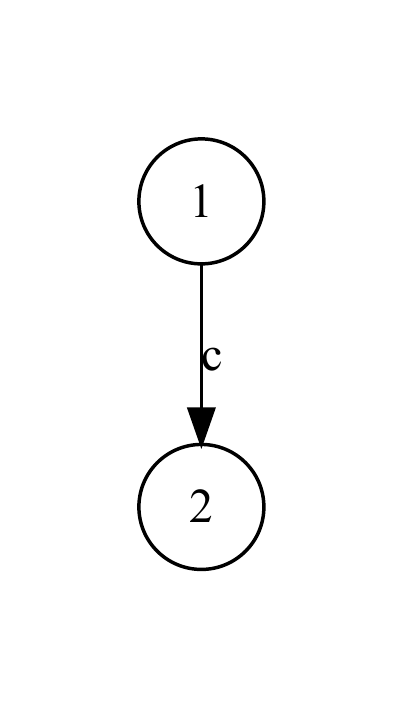}
                    \label{figure_F}
                }
                \hspace{8mm}
                \subfigure[$E \oplus F$]{
                    \includegraphics[scale=0.5]{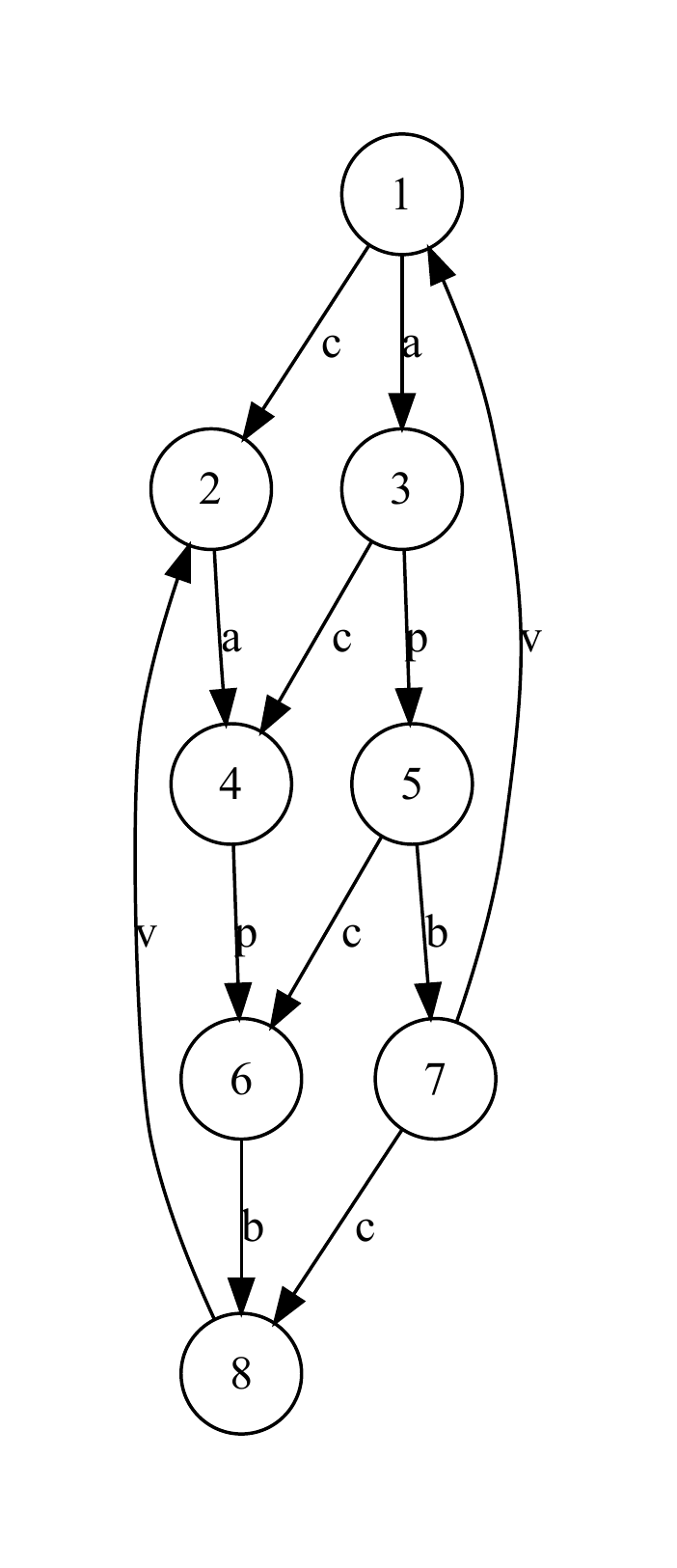}
                    \label{figure_EF}
                }
                \caption{A Counterexample}
                \label{figure_EFSystem}
            \end{figure}
        \end{example}

        In general, the only way to reduce the size of the resulting CPG is by studying the matrix $T\merry S$.
        One way would be to output the automaton from $T\merry S$ and try to find reductions afterwards.
        We decided to perform such reductions during the output process such that a unnecessarily large automaton is not generated.
        It turned out that the problems to be solved to perform these reductions are hard.
        This will be discussed in detail below.

        \newcommand{\TBDA}{\text{\sf TBD}}%
        \newcommand{\TBDNSV}{\text{\sf TBDNSV}}%
        \newcommand{\DONE}{\text{\sf DONE}}%
        \newcommand{\SUBSET}{\text{\sf SUBSET}}%
        \newcommand{\RECONSIDER}{\text{\sf RECONSIDER}}%
        \newcommand{\LNSV}{{\cal L}_{\NSV}}%
        \begin{figure}[tbp]%p
        %\centering
        \begin{minipage}{0.5\textwidth}
        \begin{algo}{OutputCPG}{}
        \TBDA \: \{startnode\}
        \TBDNSV(\ell\in \LNSV) \text{ \{array of sets; all sets initialized to $\emptyset$\}}
        \DONE \: \emptyset
         \WHILE {\TBDA \ne \emptyset\ \text{or}\ \exists \ell: \TBDNSV(\ell) \ne \emptyset}
           \IF {\TBDA \ne \emptyset}
             n \: Element(\TBDA)\quad\text{\{choose one element of set $\TBDA$\}}
             \text{print}\ n
             \FOR {\text{all edges}\ (n\to i)}
                \IF {\ell(n\to i) \in \LNSV}
                  \TBDNSV(\ell(n\to i)) \: \TBDNSV(\ell(n\to i)) \cup
                                           \qquad \qquad \qquad \qquad \qquad \ \ \{(n\to i)\}
                \ELSE
                  \TBDA \: \TBDA \cup \{i\}
                  \text{print}\ (n\to i)
                \ENDIF
                \WHILE {\exists {\cal R}: i \in {\cal R}\ \text{and}\ \exists {\cal D}: \{({\cal D}, {\cal R})\} \in \RECONSIDER}
                \quad\text{\{we have found a path back to a set of nodes}
                \quad\text{which we have used to eliminate NSV edges;}
                \quad\text{all these edges have now to be reconsidered\}}
                  \FOR {(m\to j) \in {\cal R}}
                     \TBDNSV(\ell(m\to j)) \: \TBDNSV(\ell(m\to j)) \cup
                                              \qquad \qquad \qquad \qquad \qquad \quad \{(m\to j)\}
                  \ENDFOR
                  \RECONSIDER \: \RECONSIDER \setminus \{({\cal D}, {\cal R})\}
                \ENDWHILE
             \ENDFOR
             \DONE \: \DONE \cup \{n\}; \TBDA \: \TBDA \setminus \DONE
           \ELSE\ \{ \TBDA = \emptyset\ \}
             \ell \: NonEmptyElement(\TBDNSV)\label{choose}
             \quad\text{\{choose one label with non-empty set in \TBDNSV\}}
             \SUBSET \: SmallestSubset(\TBDNSV(\ell), \DONE)
             \quad\text{\{choose smallest subset of $\TBDNSV(\ell)$ such that}
             \quad\text{subset can be reached from all nodes in set $\DONE$\}}
        %     \RECONSIDER(\DONE) \: \TBDNSV(\ell) \setminus \SUBSET
             \IF {\TBDNSV(\ell) \setminus \SUBSET \ne \emptyset}
               \RECONSIDER \: \RECONSIDER\ \cup
                              \qquad \qquad \qquad \qquad \{(\DONE,\TBDNSV(\ell) \setminus \SUBSET)\}
               \quad\text{\{remember eliminated edges;}
               \quad\text{in case we find a path back to nodes in \DONE,}
               \quad\text{we have to reconsider these edges\}}
             \ENDIF
             \FOR {(n\to i)\in \SUBSET}
               \text{print}\ (n\to i)
               \TBDA \: \TBDA \cup \{i\}
             \ENDFOR
             \TBDNSV(\ell) \: \emptyset
           \ENDIF
         \ENDWHILE
        \end{algo}
        \end{minipage}
        \caption{Output CPG}\label{outputCPG}
        \end{figure}

        Fig.~\ref{outputCPG} shows the algorithm employed for the output process in pseudo code.
        By $\ell(n\to i)$ we denote the label assigned to edge $(n\to i)$.
        In short, the algorithm records all NSV-edges and proceeds until no other edges can be processed.
        Then it chooses one label of the NSV-edges.
        From the set of all recorded edges with this label a subset is determined such that all the edges in the subset can be reached from all nodes that have been processed so long.
        This is a necessary condition, if we want to eliminate the edges outside the subset.
        Determining a minimal subset under this constraint, however, is known as the {\em Set Covering Problem\/} which is NP-hard.
        We decided to implement a greedy algorithm.
        However, it turned out that in most cases we encountered a subset of size one, which trivially is optimal.

        If no subset can be found, no edges can be eliminated.

        Concerning Ex.~\ref{example_oneNSVEdgeIsNotSufficient} we note that the reason why none of the NSV-edges can be eliminated, can be found in the presence of the loop in $E$.
        Our output algorithm traverses the CPG in such a way that we do not know in advance if a loop will be constructed later on.
        Hence our algorithm has to be aware of loops that will be constructed in the future.
        This is done by remembering eliminated edges which will be reconsidered if a suitable loop is encountered.

        In detail, if edges can be eliminated, we remember the set of eliminated edges $\cal R$ in set \RECONSIDER\ together with a copy of the current set \DONE.
        If later on we encounter a path in the CPG that reaches some nodes in this set \DONE,
        we have to reconsider our decision.
        In this case all edges in $\cal R$ are reconsidered for being present in the CPG.
        Note that several \RECONSIDER-sets can be affected if such a ``backedge'' is found.
        Note also that this reconsider mechanism handles Ex.~\ref{example_oneNSVEdgeIsNotSufficient} correctly.

        Our implementation showed that the decision which label is chosen in Line~\ref{choose} is also crucial.
        The number of edges (and nodes) being eliminated heavily depends on this choice.
        We are currently working on heuristics for this choice.

        In the following we execute the algorithm on the example of Fig.~\ref{figure_CDSystem} under the above conditions, i.e., $a$, $b$, $c$, and $d$ do not access shared variables.
        At the beginning we have $\TBDA=\{1\}$ and $\TBDNSV(a)=\TBDNSV(b)=\TBDNSV(c)=\TBDNSV(d)=\DONE=\emptyset$.
        Since \RECONSIDER-sets are not necessary in this example, we do not consider them in the following to keep things simple.

        The 1st iteration finds NSV-edges only.
        So: $\TBDNSV(a)=\{(1\to4)\}$, $\TBDNSV(c)=\{(1\to2)\}$, $\DONE=\{1\}$ and the other sets are empty.

        The 2nd iteration chooses label $a$ in Line~\ref{choose}.
        \SUBSET\ clearly is $\{(1\to4)\}$, $\TBDA=\{4\}$, and $\TBDNSV(a)=\emptyset$.

        The 3rd iteration processes Node~4 and again finds NSV-edges only.
        So: $\TBDNSV(c)=\{(1\to2), (4\to5)\}$ and $\TBDNSV(b)=\{(4\to7)\}$. $\DONE$ becomes $\{1, 4\}$.

        The 4th iteration chooses label $b$ in Line~\ref{choose}.
        Thus \SUBSET\ clearly is $\{(4\to7)\}$, $\TBDA=\{7\}$, and $\TBDNSV(b)=\emptyset$.

        The 5th iteration processes Node~7 and finds one NSV-edge labeled $c$.
        So: $\TBDNSV(c)=\{(1\to2), (4\to5), (7\to8)\}$. $\DONE$ becomes $\{1, 4, 7\}$.

        The 6th iteration handles label $c$.
        The smallest subset is found to be $\{(7\to8)\}$ since Node~7 can be reached from each of the nodes in set $\DONE=\{1, 4, 7\}$.
        Hence, edges $(1\to2)$ and $(4\to5)$ can be eliminated, i.e., they are not printed.
        So: $\TBDNSV(c)=\emptyset$ and $\TBDA=\{8\}$.

        The 7th iteration finds one NSV-edge labeled $d$.
        Thus we continue with $\TBDNSV(d)=\{(8\to9)\}$. $\DONE$ becomes $\{1, 4, 7, 8\}$.

        The 8th iteration handles label $d$.
        We obtain $\TBDNSV(d)=\emptyset$ and $\TBDA=\{9\}$.

        The 9th iteration prints Node~$9$, sets $\DONE=\{1, 4, 7, 8, 9\}$ and $\TBDA=\emptyset$.
        The algorithm terminates and the result is depicted in Fig.~\ref{figure_CDSequentialized}.

        \begin{figure}[tbh]
            \centering
            \includegraphics[scale=0.4]{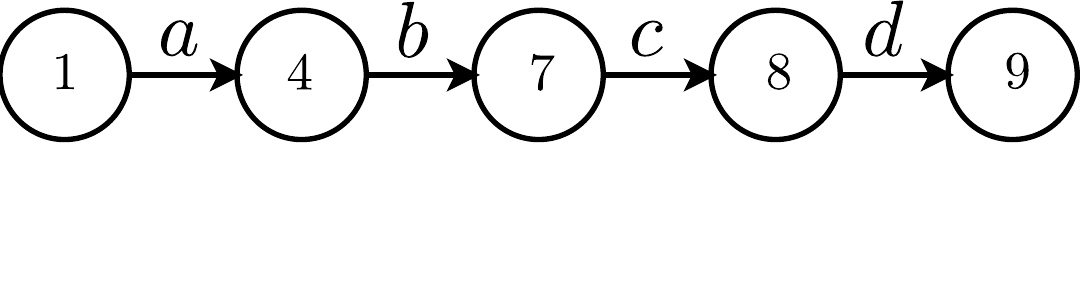}
            \caption{Sequentialized C-D-System}
            \label{figure_CDSequentialized}
        \end{figure}
    \end{subsection}
\end{section} 

%% file: dot2tex_figs/MutualExclusionOverOneSemaphoreSmall_T1.tex
\begin{tikzpicture}[scale=0.5,transform shape,>=latex,line join=bevel,]
  \pgfsetlinewidth{1bp}
\pgfsetcolor{black}
  % Edge: 1 -> 2
  \draw [->] (16bp,251.8bp) .. controls (16bp,240.07bp) and (16bp,223.94bp)  .. (16bp,200.15bp);
  \definecolor{strokecol}{rgb}{0.0,0.0,0.0};
  \pgfsetstrokecolor{strokecol}
  \draw (40.5bp,226bp) node {$T_1.p_1$};
  % Edge: 2 -> 3
  \draw [->] (16bp,167.8bp) .. controls (16bp,156.07bp) and (16bp,139.94bp)  .. (16bp,116.15bp);
  \draw (33bp,142bp) node {$T_1.a$};
  % Edge: 3 -> 4
  \draw [->] (16bp,83.804bp) .. controls (16bp,72.067bp) and (16bp,55.941bp)  .. (16bp,32.146bp);
  \draw (40.5bp,58bp) node {$T_1.v_1$};
  % Node: 1
\begin{scope}
  \definecolor{strokecol}{rgb}{0.0,0.0,0.0};
  \pgfsetstrokecolor{strokecol}
  \draw (16bp,268bp) ellipse (16bp and 16bp);
  \draw (16bp,268bp) node {$1$};
\end{scope}
  % Node: 3
\begin{scope}
  \definecolor{strokecol}{rgb}{0.0,0.0,0.0};
  \pgfsetstrokecolor{strokecol}
  \draw (16bp,100bp) ellipse (16bp and 16bp);
  \draw (16bp,100bp) node {$3$};
\end{scope}
  % Node: 2
\begin{scope}
  \definecolor{strokecol}{rgb}{0.0,0.0,0.0};
  \pgfsetstrokecolor{strokecol}
  \draw (16bp,184bp) ellipse (16bp and 16bp);
  \draw (16bp,184bp) node {$2$};
\end{scope}
  % Node: 4
\begin{scope}
  \definecolor{strokecol}{rgb}{1.0,0.0,0.0};
  \pgfsetstrokecolor{strokecol}
  \draw (16bp,16bp) ellipse (16bp and 16bp);
  \definecolor{strokecol}{rgb}{0.0,0.0,0.0};
  \pgfsetstrokecolor{strokecol}
  \draw (16bp,16bp) node {$4$};
\end{scope}
\end{tikzpicture}

%% file: dot2tex_figs/MutualExclusionOverOneSemaphoreSmall_T2.tex
\begin{tikzpicture}[scale=0.5,transform shape,>=latex,line join=bevel,]
  \pgfsetlinewidth{1bp}
\pgfsetcolor{black}
  % Edge: 1 -> 2
  \draw [->] (16bp,251.8bp) .. controls (16bp,240.07bp) and (16bp,223.94bp)  .. (16bp,200.15bp);
  \definecolor{strokecol}{rgb}{0.0,0.0,0.0};
  \pgfsetstrokecolor{strokecol}
  \draw (40.5bp,226bp) node {$T_2.p_1$};
  % Edge: 2 -> 3
  \draw [->] (16bp,167.8bp) .. controls (16bp,156.07bp) and (16bp,139.94bp)  .. (16bp,116.15bp);
  \draw (33.5bp,142bp) node {$T_2.b$};
  % Edge: 3 -> 4
  \draw [->] (16bp,83.804bp) .. controls (16bp,72.067bp) and (16bp,55.941bp)  .. (16bp,32.146bp);
  \draw (40.5bp,58bp) node {$T_2.v_1$};
  % Node: 1
\begin{scope}
  \definecolor{strokecol}{rgb}{0.0,0.0,0.0};
  \pgfsetstrokecolor{strokecol}
  \draw (16bp,268bp) ellipse (16bp and 16bp);
  \draw (16bp,268bp) node {$1$};
\end{scope}
  % Node: 3
\begin{scope}
  \definecolor{strokecol}{rgb}{0.0,0.0,0.0};
  \pgfsetstrokecolor{strokecol}
  \draw (16bp,100bp) ellipse (16bp and 16bp);
  \draw (16bp,100bp) node {$3$};
\end{scope}
  % Node: 2
\begin{scope}
  \definecolor{strokecol}{rgb}{0.0,0.0,0.0};
  \pgfsetstrokecolor{strokecol}
  \draw (16bp,184bp) ellipse (16bp and 16bp);
  \draw (16bp,184bp) node {$2$};
\end{scope}
  % Node: 4
\begin{scope}
  \definecolor{strokecol}{rgb}{1.0,0.0,0.0};
  \pgfsetstrokecolor{strokecol}
  \draw (16bp,16bp) ellipse (16bp and 16bp);
  \definecolor{strokecol}{rgb}{0.0,0.0,0.0};
  \pgfsetstrokecolor{strokecol}
  \draw (16bp,16bp) node {$4$};
\end{scope}
\end{tikzpicture}

%% file: dot2tex_figs/MutualExclusionOverOneSemaphoreSmall_CPG.tex
\begin{tikzpicture}[scale=0.5,transform shape,>=latex,line join=bevel,]
  \pgfsetlinewidth{1bp}
\pgfsetcolor{black}
  % Edge: 30 -> 31
  \draw [->] (11.413bp,84.437bp) .. controls (9.2559bp,74.204bp) and (8.1032bp,60.765bp)  .. (13bp,50bp) .. controls (16.556bp,42.181bp) and (22.898bp,35.46bp)  .. (37.66bp,24.078bp);
  \definecolor{strokecol}{rgb}{0.0,0.0,0.0};
  \pgfsetstrokecolor{strokecol}
  \draw (37.5bp,58bp) node {$T_2.v_1$};
  % Edge: 16 -> 24
  \draw [->] (88bp,167.8bp) .. controls (88bp,156.07bp) and (88bp,139.94bp)  .. (88bp,116.15bp);
  \draw (105bp,142bp) node {$T_1.a$};
  % Edge: 18 -> 25
  \draw [->] (16bp,335.8bp) .. controls (16bp,324.07bp) and (16bp,307.94bp)  .. (16bp,284.15bp);
  \draw (40.5bp,310bp) node {$T_1.v_1$};
  % Edge: 25 -> 28
  \draw [->] (16bp,251.8bp) .. controls (16bp,240.07bp) and (16bp,223.94bp)  .. (16bp,200.15bp);
  \draw (40.5bp,226bp) node {$T_2.p_1$};
  % Edge: 6 -> 7
  \draw [->] (88bp,335.8bp) .. controls (88bp,324.07bp) and (88bp,307.94bp)  .. (88bp,284.15bp);
  \draw (112.5bp,310bp) node {$T_2.v_1$};
  % Edge: 4 -> 6
  \draw [->] (88bp,419.8bp) .. controls (88bp,408.07bp) and (88bp,391.94bp)  .. (88bp,368.15bp);
  \draw (105.5bp,394bp) node {$T_2.b$};
  % Edge: 1 -> 10
  \draw [->] (37.66bp,511.92bp) .. controls (28.749bp,506.02bp) and (18.058bp,497.12bp)  .. (13bp,486bp) .. controls (9.5952bp,478.51bp) and (9.1151bp,469.74bp)  .. (11.413bp,451.56bp);
  \draw (37.5bp,478bp) node {$T_1.p_1$};
  % Edge: 7 -> 16
  \draw [->] (88bp,251.8bp) .. controls (88bp,240.07bp) and (88bp,223.94bp)  .. (88bp,200.15bp);
  \draw (112.5bp,226bp) node {$T_1.p_1$};
  % Edge: 10 -> 18
  \draw [->] (16bp,419.8bp) .. controls (16bp,408.07bp) and (16bp,391.94bp)  .. (16bp,368.15bp);
  \draw (33bp,394bp) node {$T_1.a$};
  % Edge: 28 -> 30
  \draw [->] (16bp,167.8bp) .. controls (16bp,156.07bp) and (16bp,139.94bp)  .. (16bp,116.15bp);
  \draw (33.5bp,142bp) node {$T_2.b$};
  % Edge: 1 -> 4
  \draw [->] (58.436bp,504.98bp) .. controls (63.76bp,492.56bp) and (71.425bp,474.68bp)  .. (81.62bp,450.89bp);
  \draw (97.5bp,478bp) node {$T_2.p_1$};
  % Edge: 24 -> 31
  \draw [->] (81.564bp,84.982bp) .. controls (76.24bp,72.56bp) and (68.575bp,54.675bp)  .. (58.38bp,30.888bp);
  \draw (97.5bp,58bp) node {$T_1.v_1$};
  % Node: 24
\begin{scope}
  \definecolor{strokecol}{rgb}{0.0,0.0,0.0};
  \pgfsetstrokecolor{strokecol}
  \draw (88bp,100bp) ellipse (16bp and 16bp);
  \draw (88bp,100bp) node {$24$};
\end{scope}
  % Node: 10
\begin{scope}
  \definecolor{strokecol}{rgb}{0.0,0.0,0.0};
  \pgfsetstrokecolor{strokecol}
  \draw (16bp,436bp) ellipse (16bp and 16bp);
  \draw (16bp,436bp) node {$10$};
\end{scope}
  % Node: 16
\begin{scope}
  \definecolor{strokecol}{rgb}{0.0,0.0,0.0};
  \pgfsetstrokecolor{strokecol}
  \draw (88bp,184bp) ellipse (16bp and 16bp);
  \draw (88bp,184bp) node {$16$};
\end{scope}
  % Node: 18
\begin{scope}
  \definecolor{strokecol}{rgb}{0.0,0.0,0.0};
  \pgfsetstrokecolor{strokecol}
  \draw (16bp,352bp) ellipse (16bp and 16bp);
  \draw (16bp,352bp) node {$18$};
\end{scope}
  % Node: 31
\begin{scope}
  \definecolor{strokecol}{rgb}{1.0,0.0,0.0};
  \pgfsetstrokecolor{strokecol}
  \draw (52bp,16bp) ellipse (16bp and 16bp);
  \definecolor{strokecol}{rgb}{0.0,0.0,0.0};
  \pgfsetstrokecolor{strokecol}
  \draw (52bp,16bp) node {$31$};
\end{scope}
  % Node: 30
\begin{scope}
  \definecolor{strokecol}{rgb}{0.0,0.0,0.0};
  \pgfsetstrokecolor{strokecol}
  \draw (16bp,100bp) ellipse (16bp and 16bp);
  \draw (16bp,100bp) node {$30$};
\end{scope}
  % Node: 28
\begin{scope}
  \definecolor{strokecol}{rgb}{0.0,0.0,0.0};
  \pgfsetstrokecolor{strokecol}
  \draw (16bp,184bp) ellipse (16bp and 16bp);
  \draw (16bp,184bp) node {$28$};
\end{scope}
  % Node: 1
\begin{scope}
  \definecolor{strokecol}{rgb}{0.0,0.0,0.0};
  \pgfsetstrokecolor{strokecol}
  \draw (52bp,520bp) ellipse (16bp and 16bp);
  \draw (52bp,520bp) node {$1$};
\end{scope}
  % Node: 4
\begin{scope}
  \definecolor{strokecol}{rgb}{0.0,0.0,0.0};
  \pgfsetstrokecolor{strokecol}
  \draw (88bp,436bp) ellipse (16bp and 16bp);
  \draw (88bp,436bp) node {$4$};
\end{scope}
  % Node: 7
\begin{scope}
  \definecolor{strokecol}{rgb}{0.0,0.0,0.0};
  \pgfsetstrokecolor{strokecol}
  \draw (88bp,268bp) ellipse (16bp and 16bp);
  \draw (88bp,268bp) node {$7$};
\end{scope}
  % Node: 6
\begin{scope}
  \definecolor{strokecol}{rgb}{0.0,0.0,0.0};
  \pgfsetstrokecolor{strokecol}
  \draw (88bp,352bp) ellipse (16bp and 16bp);
  \draw (88bp,352bp) node {$6$};
\end{scope}
  % Node: 25
\begin{scope}
  \definecolor{strokecol}{rgb}{0.0,0.0,0.0};
  \pgfsetstrokecolor{strokecol}
  \draw (16bp,268bp) ellipse (16bp and 16bp);
  \draw (16bp,268bp) node {$25$};
\end{scope}
\end{tikzpicture}

%% file: dot2tex_figs/MutualExclusionOverOneSemaphoreSmallUnreachable.tex
\begin{tikzpicture}[scale=0.5,transform shape,>=latex,line join=bevel,]
  \pgfsetlinewidth{1bp}
\begin{scope}
  \pgfsetstrokecolor{black}
  \definecolor{strokecol}{rgb}{1.0,1.0,1.0};
  \pgfsetstrokecolor{strokecol}
  \definecolor{fillcol}{rgb}{1.0,1.0,1.0};
  \pgfsetfillcolor{fillcol}
  \filldraw (30bp,15bp) -- (30bp,319bp) -- (417bp,319bp) -- (417bp,15bp) -- cycle;
\end{scope}
\begin{scope}
  \pgfsetstrokecolor{black}
  \definecolor{strokecol}{rgb}{1.0,1.0,1.0};
  \pgfsetstrokecolor{strokecol}
  \definecolor{fillcol}{rgb}{1.0,1.0,1.0};
  \pgfsetfillcolor{fillcol}
  \filldraw (0bp,0bp) -- (0bp,294bp) -- (294bp,294bp) -- (294bp,0bp) -- cycle;
\end{scope}
  \pgfsetcolor{black}
  % Edge: 17 -> 20
  \draw [->] (230.08bp,228.34bp) .. controls (225.56bp,217.89bp) and (219.43bp,203.7bp)  .. (210.13bp,182.16bp);
  \definecolor{strokecol}{rgb}{0.0,0.0,0.0};
  \pgfsetstrokecolor{strokecol}
  \draw (204.15bp,217bp) node {$T_2.p_1$};
  % Edge: 14 -> 15
  \draw [->] (84.268bp,111.68bp) .. controls (79.833bp,101.17bp) and (73.817bp,86.911bp)  .. (64.686bp,65.265bp);
  \draw (58.488bp,100.24bp) node {$T_2.v_1$};
  % Edge: 3 -> 12
  \draw [->] (65.754bp,229.52bp) .. controls (76.251bp,225.06bp) and (90.49bp,219.03bp)  .. (112.1bp,209.86bp);
  \draw (89.168bp,235.71bp) node {$T_1.p_1$};
  % Edge: 12 -> 20
  \draw [->] (141.69bp,196.64bp) .. controls (152.5bp,191.55bp) and (167.39bp,184.56bp)  .. (189.25bp,174.29bp);
  \draw (154.78bp,173.67bp) node {$T_1.a$};
  % Edge: 27 -> 29
  \draw [->] (270.08bp,114.26bp) .. controls (265.63bp,104.87bp) and (259.78bp,92.514bp)  .. (250.26bp,72.418bp);
  \draw (244.38bp,106bp) node {$T_2.b$};
  % Edge: 14 -> 22
  \draw [->] (105.44bp,119.82bp) .. controls (116.24bp,114.76bp) and (131.11bp,107.81bp)  .. (152.94bp,97.591bp);
  \draw (130.51bp,124.9bp) node {$T_1.a$};
  % Edge: 9 -> 17
  \draw [->] (179.45bp,270.2bp) .. controls (188.93bp,265.74bp) and (201.54bp,259.81bp)  .. (221.59bp,250.37bp);
  \draw (189.98bp,248.43bp) node {$T_1.a$};
  % Edge: 20 -> 22
  \draw [->] (196.77bp,152.68bp) .. controls (191.66bp,141.9bp) and (184.63bp,127.07bp)  .. (174.31bp,105.28bp);
  \draw (205.75bp,125.65bp) node {$T_2.b$};
  % Edge: 15 -> 23
  \draw [->] (73.281bp,43.55bp) .. controls (82.685bp,39.15bp) and (95.069bp,33.354bp)  .. (115.21bp,23.93bp);
  \draw (95.574bp,49.925bp) node {$T_1.a$};
  % Edge: 3 -> 5
  \draw [->] (43.887bp,221.41bp) .. controls (39.424bp,211.93bp) and (33.484bp,199.32bp)  .. (24.043bp,179.27bp);
  \draw (18.105bp,212.88bp) node {$T_2.b$};
  % Edge: 9 -> 12
  \draw [->] (157.43bp,262.52bp) .. controls (152.21bp,252.46bp) and (145.15bp,238.84bp)  .. (134.37bp,218.08bp);
  \draw (165.28bp,235.88bp) node {$T_2.p_1$};
  % Edge: 22 -> 29
  \draw [->] (182.43bp,84.251bp) .. controls (192.88bp,79.679bp) and (207.05bp,73.479bp)  .. (228.57bp,64.066bp);
  \draw (194.76bp,62.232bp) node {$T_1.v_1$};
  % Edge: 20 -> 27
  \draw [->] (218.26bp,159.88bp) .. controls (228.3bp,154.61bp) and (241.9bp,147.49bp)  .. (262.63bp,136.63bp);
  \draw (228.87bp,136.65bp) node {$T_1.v_1$};
  % Edge: 5 -> 14
  \draw [->] (31.546bp,156.89bp) .. controls (41.629bp,151.75bp) and (55.284bp,144.78bp)  .. (76.108bp,134.16bp);
  \draw (56.235bp,161.87bp) node {$T_1.p_1$};
  % Edge: 22 -> 23
  \draw [->] (160.05bp,76.222bp) .. controls (154.89bp,66.048bp) and (147.89bp,52.248bp)  .. (137.27bp,31.3bp);
  \draw (134bp,67.375bp) node {$T_2.v_1$};
  % Edge: 13 -> 21
  \draw [->] (382.91bp,217.32bp) .. controls (382.41bp,227.55bp) and (381.74bp,241.03bp)  .. (380.65bp,262.94bp);
  \draw (400.03bp,235.05bp) node {$T_1.a$};
  % Edge: 11 -> 19
  \draw [->] (305.06bp,213.13bp) .. controls (304.55bp,223.37bp) and (303.88bp,236.84bp)  .. (302.8bp,258.76bp);
  \draw (322.18bp,230.86bp) node {$T_1.a$};
  % Edge: 19 -> 21
  \draw [->] (318.13bp,275.74bp) .. controls (328.36bp,276.29bp) and (341.83bp,277.01bp)  .. (363.74bp,278.19bp);
  \draw (335.86bp,267.69bp) node {$T_2.b$};
  % Edge: 11 -> 13
  \draw [->] (321.99bp,197.87bp) .. controls (332.22bp,198.42bp) and (345.69bp,199.14bp)  .. (367.6bp,200.32bp);
  \draw (339.72bp,207.82bp) node {$T_2.b$};
  % Edge: 12 -> 14
  \draw [->] (119.87bp,188.79bp) .. controls (114.78bp,177.98bp) and (107.76bp,163.1bp)  .. (97.466bp,141.26bp);
  \draw (92.877bp,177.71bp) node {$T_2.b$};
  % Node: 11
\begin{scope}
  \definecolor{strokecol}{rgb}{0.0,0.0,0.0};
  \pgfsetstrokecolor{strokecol}
  \draw (306bp,197bp) ellipse (16bp and 16bp);
  \draw (305.86bp,197bp) node {$11$};
\end{scope}
  % Node: 13
\begin{scope}
  \definecolor{strokecol}{rgb}{0.0,0.0,0.0};
  \pgfsetstrokecolor{strokecol}
  \draw (384bp,201bp) ellipse (16bp and 16bp);
  \draw (383.71bp,201.19bp) node {$13$};
\end{scope}
  % Node: 12
\begin{scope}
  \definecolor{strokecol}{rgb}{0.0,0.0,0.0};
  \pgfsetstrokecolor{strokecol}
  \draw (127bp,204bp) ellipse (16bp and 16bp);
  \draw (126.86bp,203.61bp) node {$12$};
\end{scope}
  % Node: 20
\begin{scope}
  \definecolor{strokecol}{rgb}{0.0,0.0,0.0};
  \pgfsetstrokecolor{strokecol}
  \draw (204bp,167bp) ellipse (16bp and 16bp);
  \draw (203.77bp,167.46bp) node {$20$};
\end{scope}
  % Node: 21
\begin{scope}
  \definecolor{strokecol}{rgb}{0.0,0.0,0.0};
  \pgfsetstrokecolor{strokecol}
  \draw (380bp,279bp) ellipse (16bp and 16bp);
  \draw (379.85bp,279.06bp) node {$21$};
\end{scope}
  % Node: 17
\begin{scope}
  \definecolor{strokecol}{rgb}{0.0,0.0,0.0};
  \pgfsetstrokecolor{strokecol}
  \draw (237bp,243bp) ellipse (16bp and 16bp);
  \draw (236.55bp,243.33bp) node {$17$};
\end{scope}
  % Node: 23
\begin{scope}
  \definecolor{strokecol}{rgb}{0.0,0.0,0.0};
  \pgfsetstrokecolor{strokecol}
  \draw (130bp,17bp) ellipse (16bp and 16bp);
  \draw (130.02bp,17bp) node {$23$};
\end{scope}
  % Node: 19
\begin{scope}
  \definecolor{strokecol}{rgb}{0.0,0.0,0.0};
  \pgfsetstrokecolor{strokecol}
  \draw (302bp,275bp) ellipse (16bp and 16bp);
  \draw (302bp,274.87bp) node {$19$};
\end{scope}
  % Node: 32
\begin{scope}
  \definecolor{strokecol}{rgb}{0.0,0.0,0.0};
  \pgfsetstrokecolor{strokecol}
  \draw (107bp,302bp) ellipse (16bp and 16bp);
  \draw (107bp,302bp) node {$32$};
\end{scope}
  % Node: 22
\begin{scope}
  \definecolor{strokecol}{rgb}{0.0,0.0,0.0};
  \pgfsetstrokecolor{strokecol}
  \draw (167bp,91bp) ellipse (16bp and 16bp);
  \draw (167.45bp,90.803bp) node {$22$};
\end{scope}
  % Node: 26
\begin{scope}
  \definecolor{strokecol}{rgb}{0.0,0.0,0.0};
  \pgfsetstrokecolor{strokecol}
  \draw (47bp,302bp) ellipse (16bp and 16bp);
  \draw (47bp,302bp) node {$26$};
\end{scope}
  % Node: 29
\begin{scope}
  \definecolor{strokecol}{rgb}{0.0,0.0,0.0};
  \pgfsetstrokecolor{strokecol}
  \draw (243bp,58bp) ellipse (16bp and 16bp);
  \draw (243.26bp,57.641bp) node {$29$};
\end{scope}
  % Node: 15
\begin{scope}
  \definecolor{strokecol}{rgb}{0.0,0.0,0.0};
  \pgfsetstrokecolor{strokecol}
  \draw (58bp,50bp) ellipse (16bp and 16bp);
  \draw (58.453bp,50.489bp) node {$15$};
\end{scope}
  % Node: 27
\begin{scope}
  \definecolor{strokecol}{rgb}{0.0,0.0,0.0};
  \pgfsetstrokecolor{strokecol}
  \draw (277bp,129bp) ellipse (16bp and 16bp);
  \draw (277.09bp,129.05bp) node {$27$};
\end{scope}
  % Node: 3
\begin{scope}
  \definecolor{strokecol}{rgb}{0.0,0.0,0.0};
  \pgfsetstrokecolor{strokecol}
  \draw (51bp,236bp) ellipse (16bp and 16bp);
  \draw (50.709bp,235.9bp) node {$3$};
\end{scope}
  % Node: 2
\begin{scope}
  \definecolor{strokecol}{rgb}{0.0,0.0,0.0};
  \pgfsetstrokecolor{strokecol}
  \draw (302bp,32bp) ellipse (16bp and 16bp);
  \draw (302bp,32bp) node {$2$};
\end{scope}
  % Node: 5
\begin{scope}
  \definecolor{strokecol}{rgb}{0.0,0.0,0.0};
  \pgfsetstrokecolor{strokecol}
  \draw (17bp,164bp) ellipse (16bp and 16bp);
  \draw (17bp,164.32bp) node {$5$};
\end{scope}
  % Node: 9
\begin{scope}
  \definecolor{strokecol}{rgb}{0.0,0.0,0.0};
  \pgfsetstrokecolor{strokecol}
  \draw (165bp,277bp) ellipse (16bp and 16bp);
  \draw (164.96bp,277.02bp) node {$9$};
\end{scope}
  % Node: 8
\begin{scope}
  \definecolor{strokecol}{rgb}{0.0,0.0,0.0};
  \pgfsetstrokecolor{strokecol}
  \draw (242bp,302bp) ellipse (16bp and 16bp);
  \draw (242bp,302bp) node {$8$};
\end{scope}
  % Node: 14
\begin{scope}
  \definecolor{strokecol}{rgb}{0.0,0.0,0.0};
  \pgfsetstrokecolor{strokecol}
  \draw (91bp,127bp) ellipse (16bp and 16bp);
  \draw (90.623bp,126.75bp) node {$14$};
\end{scope}
\end{tikzpicture}

%% file: exampleCS.tex
\begin{section}{Client-Server Example}\label{section_exampleCS}
    We have done analysis on client-server scenarios using our lazy implementation. For the example presented here we have used clients and a semaphore of the form shown in Fig.~\ref{figure_exampleCSClient} and~\ref{figure_exampleCSSemaphore}, respectively.

\noindent
    \begin{figure*}[t]
        \centering
        \subfigure[Client]{
            \includegraphics[scale=0.4]{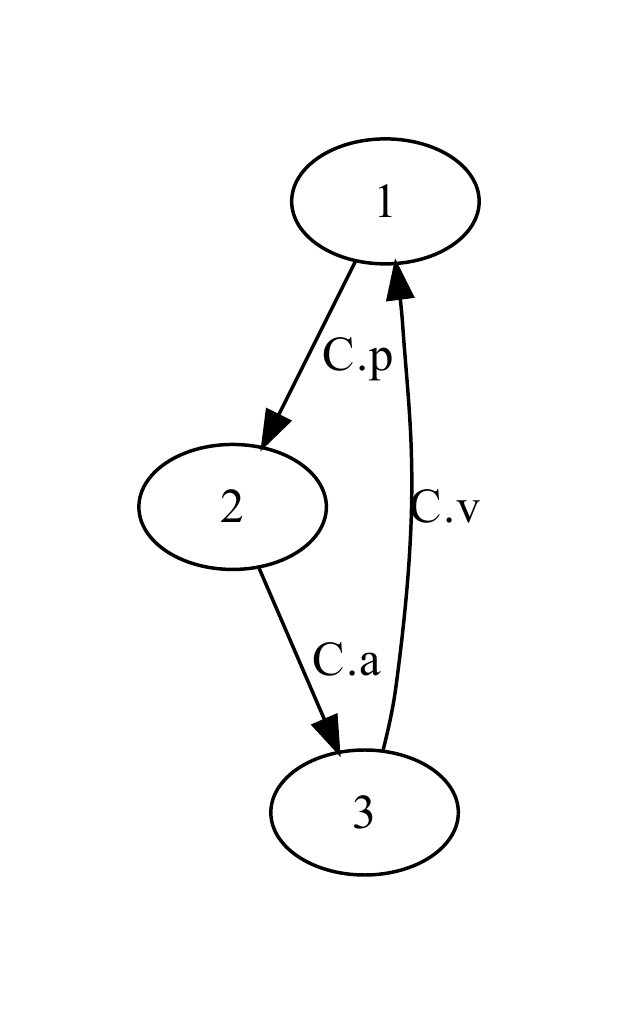}
            \label{figure_exampleCSClient}
        }
        \hspace{2mm}
        \subfigure[Semaphore]{
            \hspace{6mm}\includegraphics[scale=0.4]{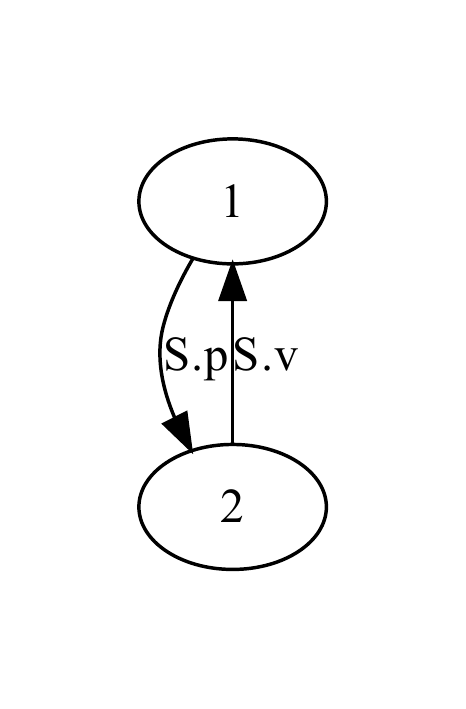}\hspace{6mm}
            \label{figure_exampleCSSemaphore}
        }
        \hspace{2mm}
        \subfigure[Statistics]{
            \begin{tabular}[tb]{| r | r | r | c | r |}
                \hline
                \rule{0in}{3ex} Clients \rule{0in}{3ex} & \rule{0in}{3ex} Nodes \rule{0in}{3ex} & \rule{0in}{3ex} Edges \rule{0in}{3ex} & \rule{0in}{3ex} Exec. Time $[s]$ \rule{0in}{3ex} & \rule{0in}{3ex} Potential Nodes \rule{0in}{3ex}\\[0.5ex] %in the last column 2x3^k where k is the number of clients
                \hline
                      1 \rule{0in}{3ex} &  \rule{0in}{1ex} 3  \rule{0in}{1ex} & \rule{0in}{1ex}  3 \rule{0in}{1ex} & \rule{0in}{1ex} 0.0013 \rule{0in}{1ex} & \rule{0in}{1ex} 6 \rule{0in}{1ex}\\
                      2 \rule{0in}{1ex} &  \rule{0in}{1ex} 5  \rule{0in}{1ex} & \rule{0in}{1ex}  6 \rule{0in}{1ex} & \rule{0in}{1ex} 0.0013 \rule{0in}{1ex} & \rule{0in}{1ex} 18 \rule{0in}{1ex}\\
                      4 \rule{0in}{1ex} &  \rule{0in}{1ex} 9  \rule{0in}{1ex} & \rule{0in}{1ex} 12 \rule{0in}{1ex} & \rule{0in}{1ex} 0.0045 \rule{0in}{1ex} & \rule{0in}{1ex} 162 \rule{0in}{1ex}\\
                      8 \rule{0in}{1ex} &  \rule{0in}{1ex} 17 \rule{0in}{1ex} & \rule{0in}{1ex} 24 \rule{0in}{1ex} & \rule{0in}{1ex} 0.0120  \rule{0in}{1ex} & \rule{0in}{1ex} 13,122 \rule{0in}{1ex}\\
                     16 \rule{0in}{1ex} &  \rule{0in}{1ex} 33 \rule{0in}{1ex} & \rule{0in}{1ex} 48 \rule{0in}{1ex} & \rule{0in}{1ex} 0.0680  \rule{0in}{1ex} & \rule{0in}{1ex} 86,093,422 \rule{0in}{1ex}\\
                     32 \rule{0in}{1ex} &  \rule{0in}{1ex} 65 \rule{0in}{1ex} & \rule{0in}{1ex} 96 \rule{0in}{1ex} & \rule{0in}{1ex} 0.4300   \rule{0in}{1ex} & \rule{0in}{1ex} 3.706 $\times 10^{15}$ \rule{0in}{1ex}\\[0.5ex]
                \hline
            \end{tabular}
            \label{table_exampleCS}
        }\\
        %\hspace{8mm}
        \subfigure[Result for 8 Clients]{
            \includegraphics[scale=0.32]{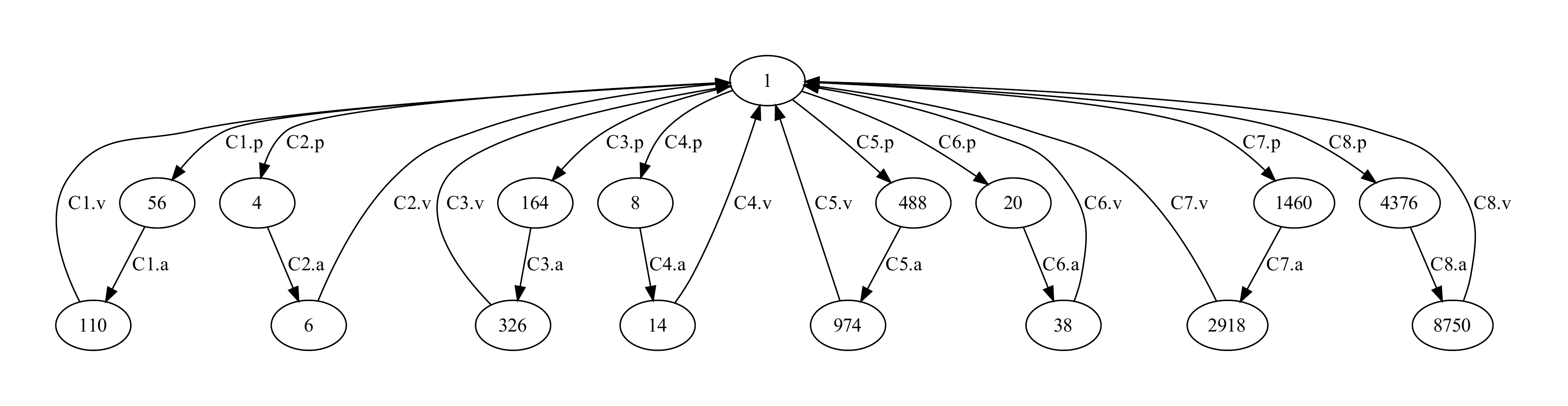}
            \label{figure_exampleCSResultEightClients}
        }
        \caption{Client-Server Example}
    \end{figure*}

    In Table~\ref{table_exampleCS} statistics for 1, 2, 4, 8, 16, and 32 clients are given. Fig.~\ref{figure_exampleCSResultEightClients} shows the resulting graph for 8 clients. The few nodes in the resulting matrix and the node IDs indicate that most nodes in the resulting matrix are superfluous. The case of 32 clients and one semaphore forms a matrix with an order of approx. $3.706 \times 10^{15}$. Our implementation generated only 65 nodes in 0.43s. In fact we observed a linear growth in the number of clients for the number of nodes and edges and for the execution time. We did our analysis on an Intel Xeon 2.8 GHz with 8GB DDR2 RAM. Note that an implementation of the matrix calculus for shared memory concurrent systems has to provide node IDs of a sufficient size. The order of $T \merry S$ can be quite big, although the resulting automaton is small.
\end{section} 

%% file: deadlock.tex
\begin{section}{Generic Proof of Deadlock Freedom}\label{section_deadlock}
    Let $S_i$ for $i\ge1$ denote binary semaphores and let their operations be denoted by $p_i$ and $v_i$.

    \begin{definition}
        Let $M=(m_{i,j})\in\cM$ denote a square matrix.
        In addition, let ${\cal P}_M=\{(i,j,r)\mid m_{i,j}=p_r\ \text{for some $r\ge1$}\}$ and
        ${\cal V}_M=\{(j,i,r)\mid m_{i,j}=v_r\ \text{for some $r\ge1$}\}$ (note the exchanged indexes $(j,i)$).
        We call $M$ {\em p-v-symmetric} iff ${\cal P}_M={\cal V}_M$.
    \end{definition}
    By definition of Kronecker sum and Kronecker product, it is easy to prove the following lemma.
    \begin{lemma}\label{lemma_p-v-sym}
        Let $M$ and $N$ be p-v-symmetric matrices.
        Then $M\oplus N$, $M\otimes N$, and $M\merry N$ are also p-v-symmetric.\qed
    \end{lemma}
    To be more specific, let $S_i=\begin{pmatrix}0&p_i\\v_i&0\end{pmatrix}$ for $i\ge1$.
    Then $S^{(r)}=\bigoplus_{i=1}^r S_i$ is p-v-symmetric.

    Now, consider the p-v-symmetric matrix
    $$
        M_k=\begin{pmatrix}0&p_1&p_2&\dots&p_k\\v_1&0&0&\dots&0\\\vdots&\vdots&\vdots&\ddots&\vdots\\v_k&0&0&\dots&0\end{pmatrix}.
    $$
    Thus $M_k^{(n)} = \bigoplus_{i=1}^n M_k$ is also p-v-symmetric.

    Now we state a theorem on deadlock freedom.
    \begin{theorem}\label{theorem_deadlockFree}
        Let $P=M_k^{(n)} \circ S^{(k)}$ be the matrix of a $n$-threaded program with $k$ binary semaphores, where $M_k^{(n)}$ and $S^{(k)}$ are defined above.
        Then the program is deadlock free.
    \end{theorem}
    \begin{proof}
        By definition and Lemma~\ref{lemma_p-v-sym} $P$ is p-v-symmetric.
        By Corollary~\ref{corollary_zeroline} a deadlock manifests itself by a zero line, say $\ell$, in matrix $P$.
        Since $P$ is p-v-symmetric, column $\ell$ does only contain zeroes.
        Hence line $\ell$ is unreachable in the underlying automaton.

        This clearly holds for all zero lines in $P$ and thus the program is deadlock free.\qed
    \end{proof}

    For counting semaphores we obtain matrices of the following type
    $$
        \begin{pmatrix}
            0 & p & 0 & \cdots & 0 & 0\\
            v & 0 & p & \cdots & 0 & 0\\
            0 & v & 0 & \cdots & 0 & 0\\
            \vdots & \vdots & \vdots & \ddots & \vdots & \vdots\\
            0 & 0 & 0 & \cdots & 0 & p\\
            0 & 0 & 0 & \cdots & v & 0
        \end{pmatrix}
    $$
    which clearly is p-v-symmetric.
    Thus a similar theorem holds if counting sema\-phores are used instead of binary ones.

    A short reflection shows that if we allow $M_k$ to contain additional entries and non-zero lines and columns
    which do not contain $p$s and $v$s, the system is still deadlock free.
    So, we have derived a very powerful criterion to ensure deadlock freedom for a large class of programs, namely p-v-symmetric programs.

    Concerning the example in Section~\ref{section_exampleCS} we note that if edges labeled $a$ are removed from the clients, we obtain p-v-symmetric matrices.
    Thus this simple client-server system is deadlock free for an arbitrary number of clients.
    If we reinsert edges labeled $a$ into the clients, no zero lines and columns appear (as
    noted above), so that the system is still deadlock free for an arbitrary number of clients.

    Theorem~\ref{theorem_deadlockFree} may be compared to the results of~\cite{EN:95,CTTV:04}, where for homogenous token passing rings it is proved that checking correctness properties can be reduced to rings of small sizes.
\end{section} 

%% file: example.tex
\begin{section}{A Data Race Example}\label{section_example}
    We give an example, where a programmer is supposed to have used synchronization primitives in a wrong way. The program consisting of two threads, namely $T_1$ and $T_2$, and a semaphore $s$ is given in Fig.~\ref{figure_exampleProgram}. We assume that $sv=0$ at program start. It is supposed that the program delivers $sv=2$ when it terminates. Both threads in the program access the shared variable $sv$. The variables $r$ and $t$ are local to the corresponding threads. The programmer inadvertently has placed line 1 in front of line 2 in $T_2$.

    \begin{figure}[t]
        \centering
        \begin{minipage}{3cm}
            \begin{algo}{$T_1$}{}
                s.p             \text{\qquad \quad \quad \hspace{10pt}  \{edge T1.p\}}
                    r \: sv     \text{\quad \quad \hspace{11pt} \{edge a\}}
                    r \: r + 1  \text{\quad \hspace{9pt} \{edge b\}}
                    sv \: r     \text{\quad \quad \hspace{11pt}  \{edge b\}}
                s.v             \text{\qquad \quad \quad \hspace{10pt} \{edge T1.v\}}
            \end{algo}
            \begin{algo}{$T_2$}{}
                t \: sv         \text{\quad \quad \hspace{12pt} \{edge c\}}
                s.p             \text{\qquad \quad \quad \hspace{10pt} \{edge T2.p\}}
                    t \: t + 1  \text{\quad \hspace{11pt} \{edge d\}}
                    sv \: t     \text{\quad \quad \hspace{12pt} \{edge d\}}
                s.v             \text{\qquad \quad \quad \hspace{10pt} \{edge T2.v\}}
            \end{algo}
        \end{minipage}
        \caption{Example Program}\label{figure_exampleProgram}
    \end{figure}

    \begin{figure}[t]
        \centering
        \subfigure[$T_1$]{
            \includegraphics[scale=0.4]{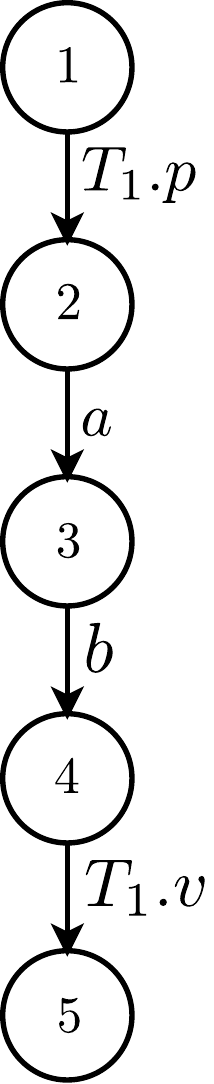}
            \label{figure_exampleRCFGT1}
        }
        \hspace{8mm}
        \subfigure[$T_2$]{
            \includegraphics[scale=0.4]{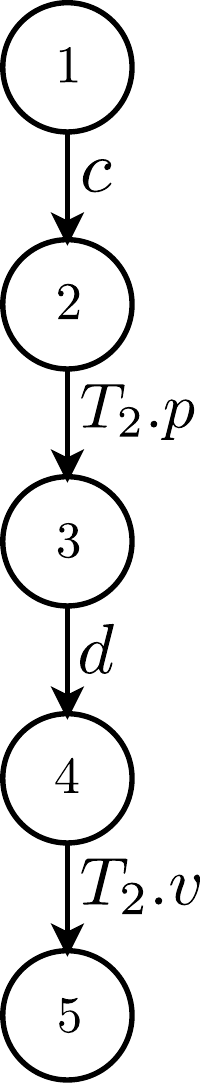}
            \label{figure_exampleRCFGT2}
        }
        \caption{RCFGs after Edge Splitting}
        \label{figure_exampleRCFGs}
    \end{figure}

    After edge splitting we get the RCFGs depicted in Fig.~\ref{figure_exampleRCFGs}. As usual the semaphore looks like Fig.~\ref{figure_semaphoreBinary}.
    The corresponding matrices are
    $$%\begin{eqnarray*}
        T_1=
        \begin{pmatrix}
            0 & T_1.p & 0 & 0 & 0     \\
            0 & 0     & a & 0 & 0     \\
            0 & 0     & 0 & b & 0     \\
            0 & 0     & 0 & 0 & T_1.v \\
            0 & 0     & 0 & 0 & 0     \\
        \end{pmatrix}
        \mbox{ and }
        T_2=
        \begin{pmatrix}
            0 & c & 0     & 0 & 0     \\
            0 & 0 & T_2.p & 0 & 0     \\
            0 & 0 & 0     & d & 0     \\
            0 & 0 & 0     & 0 & T_2.v \\
            0 & 0 & 0     & 0 & 0     \\
        \end{pmatrix}.
    $$%\end{eqnarray*}

    Although the following matrices are not computed by our lazy implementation, we give them here to allow the reader to see a complete example.
    %Let $I^x_n$ denote a matrix similar to the identity matrix of order $n$ with $x$ (instead of $1$) on the main diagonal and zeros elsewhere.
    To enable a concise presentation we define the following submatrices of order five:

    \begin{eqnarray*}
        H&=&
        \begin{pmatrix}
            0 & c & 0     & 0 & 0\\
            0 & 0 & T_2.p & 0 & 0\\
            0 & 0 & 0     & d & 0\\
            0 & 0 & 0     & 0 & T_2.v\\
            0 & 0 & 0     & 0 & 0\\
        \end{pmatrix},
        I=
        \begin{pmatrix}
            T_1.p & 0     & 0     & 0     & 0\\
            0     & T_1.p & 0     & 0     & 0\\
            0     & 0     & T_1.p & 0     & 0\\
            0     & 0     & 0     & T_1.p & 0\\
            0     & 0     & 0     & 0     & T_1.p\\
        \end{pmatrix},\\
        J&=&a \cdot I_{5},
        %\begin{pmatrix}
        %    a & 0 & 0 & 0 & 0\\
        %    0 & a & 0 & 0 & 0\\
        %    0 & 0 & a & 0 & 0\\
        %    0 & 0 & 0 & a & 0\\
        %    0 & 0 & 0 & 0 & a\\
        %\end{pmatrix},
        K=b \cdot I_{5}
        %\begin{pmatrix}
        %    b & 0 & 0 & 0 & 0\\
        %    0 & b & 0 & 0 & 0\\
        %    0 & 0 & b & 0 & 0\\
        %    0 & 0 & 0 & b & 0\\
        %    0 & 0 & 0 & 0 & b\\
        %\end{pmatrix}
        , \mbox{ and }
        L=
        \begin{pmatrix}
            T_1.v & 0     & 0     & 0     & 0\\
            0     & T_1.v & 0     & 0     & 0\\
            0     & 0     & T_1.v & 0     & 0\\
            0     & 0     & 0     & T_1.v & 0\\
            0     & 0     & 0     & 0     & T_1.v\\
        \end{pmatrix}.
    \end{eqnarray*}

    Now, we get $T=T_1 \oplus T_2$, a matrix of order 25, consisting of the submatrices defined above and zero matrices of order five (instead of $Z_5$ simply denoted by $0$).
    $$T=\begin{pmatrix}
            H & I & 0 & 0 & 0\\
            0 & H & J & 0 & 0\\
            0 & 0 & H & K & 0\\
            0 & 0 & 0 & H & L\\
            0 & 0 & 0 & 0 & H\\
        \end{pmatrix}.
    $$
    To shorten the presentation of $P=T \merry S$ we define the following submatrices of order ten:
    { %\scriptsize
    \begin{eqnarray*}
        U&=&
        \begin{pmatrix}
          % 1   2   3   4   5   6       7   8   9       10
            0 & 0 & c & 0 & 0 & 0     & 0 & 0 & 0     & 0\\ %1
            0 & 0 & 0 & c & 0 & 0     & 0 & 0 & 0     & 0\\ %2
            0 & 0 & 0 & 0 & 0 & T_2.p & 0 & 0 & 0     & 0\\ %3
            0 & 0 & 0 & 0 & 0 & 0     & 0 & 0 & 0     & 0\\ %4
            0 & 0 & 0 & 0 & 0 & 0     & d & 0 & 0     & 0\\ %5
            0 & 0 & 0 & 0 & 0 & 0     & 0 & d & 0     & 0\\ %6
            0 & 0 & 0 & 0 & 0 & 0     & 0 & 0 & 0     & 0\\ %7
            0 & 0 & 0 & 0 & 0 & 0     & 0 & 0 & T_2.v & 0\\ %8
            0 & 0 & 0 & 0 & 0 & 0     & 0 & 0 & 0     & 0\\ %9
            0 & 0 & 0 & 0 & 0 & 0     & 0 & 0 & 0     & 0\\ %10
        \end{pmatrix},
        V=
        \begin{pmatrix}
          % 1   2       3   4       5   6   7   8   9   10
            0 & T_1.p & 0 & 0     & 0 & 0     & 0 & 0     & 0 & 0    \\ %1
            0 & 0     & 0 & 0     & 0 & 0     & 0 & 0     & 0 & 0    \\ %2
            0 & 0     & 0 & T_1.p & 0 & 0     & 0 & 0     & 0 & 0    \\ %3
            0 & 0     & 0 & 0     & 0 & 0     & 0 & 0     & 0 & 0    \\ %4
            0 & 0     & 0 & 0     & 0 & T_1.p & 0 & 0     & 0 & 0    \\ %5
            0 & 0     & 0 & 0     & 0 & 0     & 0 & 0     & 0 & 0    \\ %6
            0 & 0     & 0 & 0     & 0 & 0     & 0 & T_1.p & 0 & 0    \\ %7
            0 & 0     & 0 & 0     & 0 & 0     & 0 & 0     & 0 & 0    \\ %8
            0 & 0     & 0 & 0     & 0 & 0     & 0 & 0     & 0 & T_1.p\\ %9
            0 & 0     & 0 & 0     & 0 & 0     & 0 & 0     & 0 & 0    \\ %10
        \end{pmatrix},%\\
        \end{eqnarray*}
        \begin{eqnarray*}
        W&=&a \cdot I_{10},\mbox{ }
        %\begin{pmatrix}
        %  % 1   2   3   4   5   6   7   8   9   10
        %    a & 0 & 0 & 0 & 0 & 0 & 0 & 0 & 0 & 0\\ %1
        %    0 & a & 0 & 0 & 0 & 0 & 0 & 0 & 0 & 0\\ %2
        %    0 & 0 & a & 0 & 0 & 0 & 0 & 0 & 0 & 0\\ %3
        %    0 & 0 & 0 & a & 0 & 0 & 0 & 0 & 0 & 0\\ %4
        %    0 & 0 & 0 & 0 & a & 0 & 0 & 0 & 0 & 0\\ %5
        %    0 & 0 & 0 & 0 & 0 & a & 0 & 0 & 0 & 0\\ %6
        %    0 & 0 & 0 & 0 & 0 & 0 & a & 0 & 0 & 0\\ %7
        %    0 & 0 & 0 & 0 & 0 & 0 & 0 & a & 0 & 0\\ %8
        %    0 & 0 & 0 & 0 & 0 & 0 & 0 & 0 & a & 0\\ %9
        %    0 & 0 & 0 & 0 & 0 & 0 & 0 & 0 & 0 & a\\ %10
        %\end{pmatrix},\\
        X=b \cdot I_{10}
        %\begin{pmatrix}
        %  % 1   2   3   4   5   6   7   8   9   10
        %    b & 0 & 0 & 0 & 0 & 0 & 0 & 0 & 0 & 0\\ %1
        %    0 & b & 0 & 0 & 0 & 0 & 0 & 0 & 0 & 0\\ %2
        %    0 & 0 & b & 0 & 0 & 0 & 0 & 0 & 0 & 0\\ %3
        %    0 & 0 & 0 & b & 0 & 0 & 0 & 0 & 0 & 0\\ %4
        %    0 & 0 & 0 & 0 & b & 0 & 0 & 0 & 0 & 0\\ %5
        %    0 & 0 & 0 & 0 & 0 & b & 0 & 0 & 0 & 0\\ %6
        %    0 & 0 & 0 & 0 & 0 & 0 & b & 0 & 0 & 0\\ %7
        %    0 & 0 & 0 & 0 & 0 & 0 & 0 & b & 0 & 0\\ %8
        %    0 & 0 & 0 & 0 & 0 & 0 & 0 & 0 & b & 0\\ %9
        %    0 & 0 & 0 & 0 & 0 & 0 & 0 & 0 & 0 & b\\ %10
        %\end{pmatrix}
        , \mbox{ and }
        Y=
        \begin{pmatrix}
          % 1       2   3       4   5       6   7       8   9       10
            0     & 0 & 0     & 0 & 0     & 0 & 0     & 0 & 0     & 0\\ %1
            T_1.v & 0 & 0     & 0 & 0     & 0 & 0     & 0 & 0     & 0\\ %2
            0     & 0 & 0     & 0 & 0     & 0 & 0     & 0 & 0     & 0\\ %3
            0     & 0 & T_1.v & 0 & 0     & 0 & 0     & 0 & 0     & 0\\ %4
            0     & 0 & 0     & 0 & 0     & 0 & 0     & 0 & 0     & 0\\ %5
            0     & 0 & 0     & 0 & T_1.v & 0 & 0     & 0 & 0     & 0\\ %6
            0     & 0 & 0     & 0 & 0     & 0 & 0     & 0 & 0     & 0\\ %7
            0     & 0 & 0     & 0 & 0     & 0 & T_1.v & 0 & 0     & 0\\ %8
            0     & 0 & 0     & 0 & 0     & 0 & 0     & 0 & 0     & 0\\ %9
            0     & 0 & 0     & 0 & 0     & 0 & 0     & 0 & T_1.v & 0\\ %10
        \end{pmatrix}.
    \end{eqnarray*}
    }

    With the help of zero matrices of order ten we can state the program's matrix

    \begin{figure}[tbh]
        \centering
        \hrule
        \vspace{4mm}
        \includegraphics[scale=0.6]{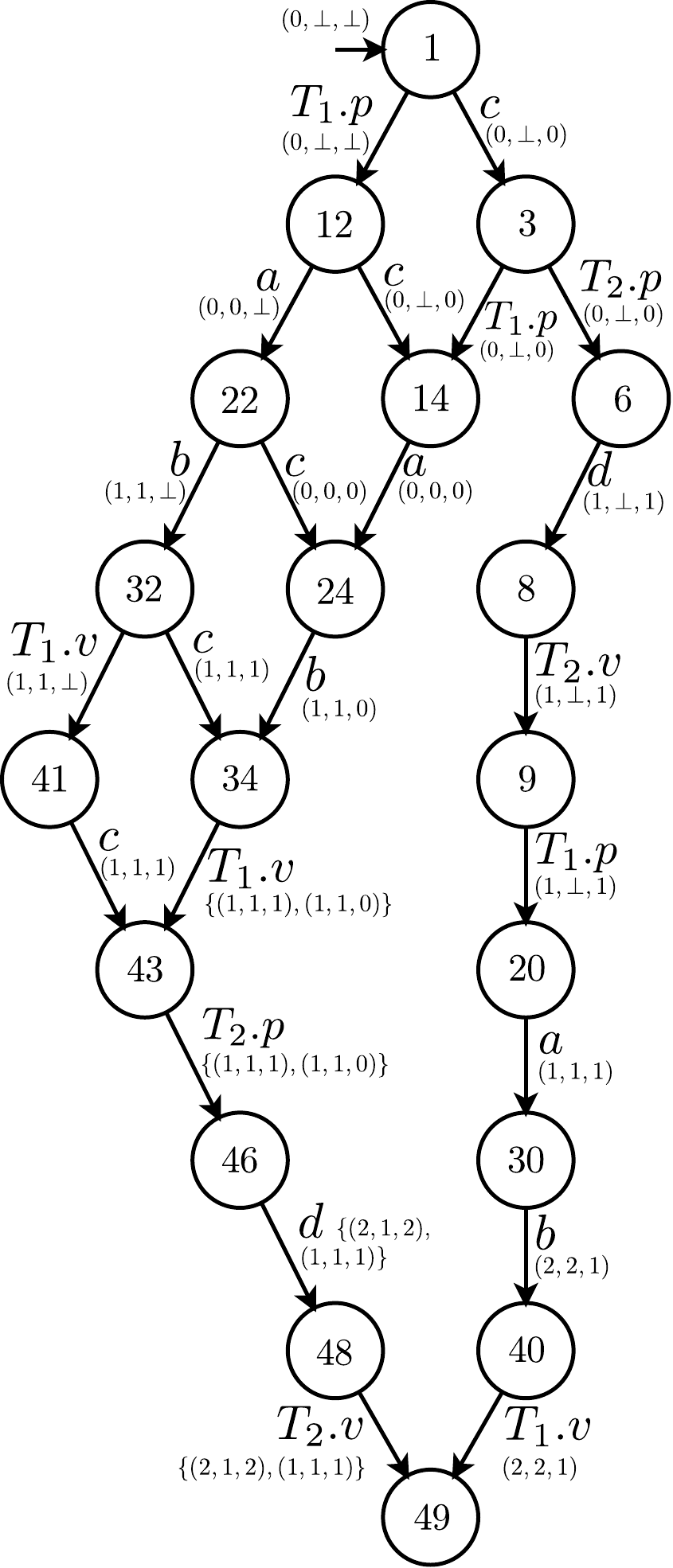}
        \caption{Resulting CPG}
        \label{figure_exampleCPG}
    \end{figure}

    \begin{eqnarray*}
        P=T\merry S=T \merry
        \begin{pmatrix}
            0 & p\\
            v & 0\\
        \end{pmatrix}=
        \begin{pmatrix}
            U & V & 0 & 0 & 0\\
            0 & U & W & 0 & 0\\
            0 & 0 & U & X & 0\\
            0 & 0 & 0 & U & Y\\
            0 & 0 & 0 & 0 & U\\
        \end{pmatrix}.
    \end{eqnarray*}

    Matrix $P$ has order 50. The corresponding CPG is shown in Fig.~\ref{figure_exampleCPG}. The lazy implementation computes only these 19 nodes. Due to synchronization the other parts are not reachable. In addition to the usual labels we have add a set of tuples to each edge in the CPG of Fig.~\ref{figure_exampleCPG}. Tuple $(x,y,z)$ denotes values of variables, such that $sv=x$, $r=y$ and $t=z$. We use $\bot$ to refer to an undefined value. A tuple shows the values after the basic block on the corresponding edge has been evaluated. The entry node of the CPG is Node $1$. At program start we have the variable assignment $(0,\bot,\bot)$. At Node $49$ we result in the set of tuples $\{(1,1,1),(2,1,2),(2,2,1)\}$. Due to the interleavings different tuples may occur at join nodes. This we reflect by a set of tuples. As stated above the program is supposed to deliver $sv=2$. Thus the tuple $(1,1,1)$ shows that the program is erroneous. The error is caused by a data race between the edges $c$ of thread $T_2$ and the edges $a$ and $b$ of thread $T_1$.
\end{section} 

%% file: empirical.tex
\begin{section}{Empirical Data}\label{section_empiricalData}
    In Sec.~\ref{section_exampleCS} we already gave some empirical data concerning client-server examples. In this section we give empirical data for ten additional examples.

    Let $o(P)$ and $o(C)$ refer to the order of the adjacency matrix $P$, which is not computed by our lazy implementation, and the order of the adjacency matrix $C$ of the resulting CPG, respectively. In addition $k$ and $r$ refer to the number of threads and the number of semaphores, respectively.

    \begin{table}[th]
        \centering
        \begin{tabular}[t]{| c | r | r | r@{}c@{}l | r | r@{}c@{}l|}
            \hline
            \rule{0in}{3ex} k \rule{0in}{2ex}&\rule{0in}{2ex} r \rule{0in}{2ex}&\rule{0in}{2ex}$o(P)$\rule{0in}{2ex}&\multicolumn{3}{c|}{\rule{0in}{2ex}$\sqrt{(o(P))}$\rule{0in}{2ex}}& \rule{0in}{2ex}$o(C)$\rule{0in}{2ex}&\multicolumn{3}{c|}{\rule{0in}{2ex}Runtime [s]\rule{0in}{3ex}}\\[0.5ex]
            \hline
            2 & 4 & \rule{0in}{3ex}256\rule{0in}{2ex} & \rule{0in}{2ex}16&,&00\rule{0in}{2ex}    & \rule{0in}{2ex}12\rule{0in}{2ex}   & \rule{0in}{2ex}0&,&03\\%bsp81od
            3 & 5 & \rule{0in}{2ex}4800\rule{0in}{2ex} & \rule{0in}{2ex}69&,&28\rule{0in}{2ex}    & \rule{0in}{2ex}30\rule{0in}{2ex}   & \rule{0in}{2ex}0&,&097542\rule{0in}{2ex}\\%bsp55od
            4 & 6 & \rule{0in}{2ex}124416\rule{0in}{2ex} & \rule{0in}{2ex}352&,&73\rule{0in}{2ex}   & \rule{0in}{2ex}98\rule{0in}{2ex}   & \rule{0in}{2ex}0&,&48655\\%bsp89od
            3 & 6 & \rule{0in}{2ex}75264\rule{0in}{2ex} & \rule{0in}{2ex}274&,&34\rule{0in}{2ex}   & \rule{0in}{2ex}221\rule{0in}{2ex}  & \rule{0in}{2ex}1&,&057529\\%Andi
            4 & 7 & \rule{0in}{2ex}614400\rule{0in}{2ex} & \rule{0in}{2ex}783&,&84\rule{0in}{2ex}   & \rule{0in}{2ex}338\rule{0in}{2ex}  & \rule{0in}{2ex}2&,&537082\\%bsp86od
            4 & 8 & \rule{0in}{2ex}1536000\rule{0in}{2ex} & \rule{0in}{2ex}1239&,&35\rule{0in}{2ex}  & \rule{0in}{2ex}277\rule{0in}{2ex}  & \rule{0in}{2ex}2&,&566587\rule{0in}{2ex}\\%bsp61od
            4 & 8 & \rule{0in}{2ex}737280\rule{0in}{2ex} & \rule{0in}{2ex}858&,&65\rule{0in}{2ex}   & \rule{0in}{2ex}380\rule{0in}{2ex}  & \rule{0in}{2ex}3&,&724364\\%bsp168od
            4 & 13 & \rule{0in}{2ex}298721280\rule{0in}{2ex} & \rule{0in}{2ex}17283&,&56\rule{0in}{2ex} & \rule{0in}{2ex}2583\rule{0in}{2ex} & \rule{0in}{2ex}96&,&024073\rule{0in}{2ex}\\%Andi2überhol
            4 & 11 & \rule{0in}{2ex}55050240\rule{0in}{2ex} & \rule{0in}{2ex}7419&,&58\rule{0in}{2ex}  & \rule{0in}{2ex}3908\rule{0in}{2ex} & \rule{0in}{2ex}146&,&81\\%Andi2 (gibt es nicht mehr im Verzeichnis, ist jetzt Andi2überhol)
            5 & 6 & \rule{0in}{2ex}14929920\rule{0in}{2ex} & \rule{0in}{2ex}3863&,&93\rule{0in}{2ex}  & \rule{0in}{2ex}7666\rule{0in}{2ex} & \rule{0in}{2ex}309&,&371395\rule{0in}{2ex}\\[1ex]%bsp76od
            \hline
        \end{tabular}
        \caption{Empirical Data}
        \label{table_empiricalData}
    \end{table}

    In the following we use the data depicted in Table~\ref{table_empiricalData}.\footnote{We did our analysis on an Intel Pentium D 3.0 GHz machine with 1GB DDR RAM running CentOS 5.6.} The numbers in the third column are rounded to two decimal places. As a first observation we note that except for one example all values of $o(C)$ are smaller as the corresponding values of $\sqrt{(o(P))}$.
    In addition, the runtime of our implementation shows a strong correlation to the order $o(C)$ of the adjacency matrix $C$ of the generated CPG with a Pearson product-moment correlation coefficient of
    0,9990277130. %using gnumeric: =correl(D3:D12;F3:F12) %this is exactly the value when applying PEARSON in Open Office
    %0,998056371 %using excel: =BESTIMMTHEITSMASS(D3:D12;F3:F12)
    %0,999027712955541  %using open office: PEARSON(D3:D12;F3:F12) this is the Pearson product-moment correlation coefficient
    In contrast the values of the theoretical order $o(P)$ of the resulting adjacency matrix $P$ correlates to the runtime only with a correlation coefficient of 0.2370050995.\footnote{Both correlation coefficients are rounded to ten decimal places.}

    This observations show that the runtime complexity does not depend on the order $o(P)$ which
    %in the worst-case
    grows exponentially in the number of threads. We conclude this section by stating that the collected data give strong indication that the runtime complexity of our approach is linear in the number of nodes present in the resulting CPG.
\end{section} 

%% file: related.tex
\begin{section}{Related Work}\label{section_related}
    %With this paper we have established an approach for constructing CPGs for shared memory concurrent programs.
    Probably the closest work to ours was done by Buchholz and Kemper~\cite{BK:02}. It differs from our work as stated in the following. We establish a framework for analyzing multithreaded shared memory concurrent systems which forms a basis for studying various properties of the program. Different techniques including dataflow analysis (e.g.~\cite{RP:86,RP:88,SGL:98,KU:76})
    %, symbolic analysis e.g.~\cite{BSB:08},
    and model checking (e.g.~\cite{CGP:99,GG:08} to name only a few) can be applied to the generated CPGs. In this paper we use our approach in order to prove deadlock freedom. Buchholz and Kemper worked on generating reachability sets in composed automata. Our approach uses CFGs and semaphores to model shared memory concurrent programs. Buchholz and Kemper use it for describing networks of synchronized automata. Both approaches employ Kronecker algebra. An additional difference is that we propose optimizations concerning the handling of edges not accessing shared variables and lazy evaluation of the matrix entries.
    %We do optimizations concerning edges which do not interfere.
%
    %In~\cite{GM:11} Garg and Madhusudan present sequentialization of concurrent programs. It is shown that if for a concurrent program a compositional proof exists, then it can be translated to a sequential program. Nevertheless, a drawback of the approach is that it generates recursive programs, even when the concurrent program contains no recursion. Our CPG can be seen as a sequentialization (cf. Fig.~\ref{figure_CDSequentialized} and Subsect.~\ref{subsection_optimizationForNSV}) of a concurrent program in the sense of~\cite{GM:11}.

    In~\cite{GG:08} Ganai and Gupta studied modeling concurrent systems for bounded model checking (BMC). Somehow similar to our approach the concurrent system is modeled lazily. In contrast our approach does not need temporal logic specifications like LTL for proving deadlock freedom for p-v-symmetric programs but on the other hand our approach may suffer from false positives. Like all BMC approaches~\cite{GG:08} has the drawback that it can only show correctness within a bounded number of $k$ steps.

    Kahlon et al. propose a framework for static analysis of concurrent programs in~\cite{KSG:09}. Partial order reduction and synchronization constraints are used to reduce thread interleavings. In order to gain further reductions abstract interpretation is applied.

    In~\cite{RGG+:95} a model checking tool is presented that builds up a system gradually, at each stage compressing the subsystems to find an equivalent CSP process with many less states. With this approach systems of exponential size ($\geq 10^{20}$) can be model checked successfully. This can be compared to our client-server example in Sect.~\ref{section_exampleCS}, where matrices of exponential size can be handled in linear time.

    Although not closely related we recognize the work done in the field of \emph{stochastic automata networks} (SAN) which is based on the work of Plateau~\cite{Pla:85} and in the field of \emph{generalized stochastic petri nets} (GSPN) (e.g.~\cite{CM:99}) as related work. Compared to ours these fields are completely different. Nevertheless, basic operators are shared and some properties influenced this paper.
    %We note that the approach of~\cite{Pla:85} may profit from a lazy implementation.
\end{section} 

%% file: conclusion.tex
\begin{section}{Conclusion}\label{section_conclusion}
    We established a framework for analyzing multithreaded shared memory concurrent systems which forms a basis for studying various properties of programs. Different techniques including dataflow analysis
    %e.g.~\cite{RP:86,RP:88,SGL:98,KU:76}
    %symbolic analysis (e.g.~\cite{BSB:08}),
    and model checking
    %e.g.~\cite{CGP:99,GG:08} to name only a few)
    can be applied to CPGs.
    In addition, the structure of the matrices can be used to prove properties of the underlying program for an arbitrary number of threads. In this paper we used CPGs in order to prove deadlock freedom for the large class of p-v-symmetric programs.

    Furthermore, we proved that in general CPGs can be represented by sparse matrices. Hence the number of entries in the matrices is linear in their number of lines. Thus efficient algorithms can be applied to CPGs.

    We proposed two major optimizations. First, if the program contains a lot of synchronization, only a very small part of the CPG is reachable and, due to a lazy implementation of the matrix operations, only this part is computed. Second, if the program has only little synchronization, many edges not accessing shared variables will be present, which are reduced during the output process of the CPG. Both optimizations speed up processing significantly and show that this approach is very promising.

    We gave examples for both, the lazy implementation and how we are able to prove deadlock freedom.

    The first results of our approach (such as Theorem~\ref{theorem_deadlockFree}) and the performance of our prototype implementation are very promising. Further research is needed to generalize Theorem~\ref{theorem_deadlockFree} in order to handle systems similar to the Dining Philosophers problem. In addition, details on how to perform (complete and sound) dataflow analysis on CPGs have to be studied.
\end{section} 

%% file: references.tex
\bibliographystyle{abbrv} %default
%\bibliographystyle{alphaabb}
%\bibliographystyle{plainabb}
%\bibliography{/ABT/USER/staff/blieb/tex/bib/oo,/ABT/USER/staff/blieb/tex/bib/wcp,/ABT/USER/staff/blieb/tex/bib/thi,/ABT/USER/staff/blieb/tex/bib/corr,/ABT/USER/staff/blieb/tex/bib/ada}
%\bibliography{../bib/oo,../bib/wcp,../bib/thi,../bib/corr,../bib/ada,../bib/dataflow,../bib/icaroos,../bib/sym}
%\subsubsection{References}
%\renewcommand\refname{}
%\vspace{-12mm}
\bibliography{ada} %default
%\bibliographystyleapproach{abbrv}
%\bibliographyapproach{ada}

%\begin{thebibliography}{10}
%\softraggedright
%place entries from bbl file here
%\end{thebibliography} 